 \documentclass[sigplan,screen]{acmart}
\renewcommand\footnotetextcopyrightpermission[1]{} 
\pdfoutput=1
\setcopyright{none} 

\usepackage{hyperref}
\usepackage[utf8]{inputenc}
\usepackage{color}
\usepackage{xspace}
\usepackage[inference]{semantic}
\usepackage{verbatim}
\usepackage{amssymb}
\usepackage[normalem]{ulem}
\usepackage{xspace}
\usepackage{amsthm}
\usepackage{multirow}
\usepackage{showexpl} 
\usepackage{textcomp}
\usepackage{makecell}
\usepackage{stmaryrd}
\usepackage{pgfplots}
\pgfplotsset{width=7cm,compat=1.8}
\usepackage{float}
\usepackage{flushend}
\usepackage{subcaption}
\usepackage{balance}

\newtheorem{definition}{Definition}
\newtheorem{lemma}{Lemma}
\newtheorem{assumption}{Assumption}

\def\codesize{\normalsize}
\definecolor{rubystr}{gray}{0.1}
\definecolor{gray_ulisses}{gray}{0.5}
\definecolor{preto_ulisses}{gray}{0.0}

\lstdefinelanguage{ruby}{
	basicstyle=\sf\small\fontseries{m}\selectfont,
        columns=flexible,
	showstringspaces=false,
	literate={ {->}{{$\rightarrow$}}2
               {&&}{{$\land$}}2
               {||}{{$\lor$}}2
               {=>}{{$\Rightarrow$}}2
               {<<}{{\guillemotleft}}1 
               {>>}{{\guillemotright}}1
               {~>}{{$\hookrightarrow$}}2
               {|->}{{$\mapsto$}}1
             },
       escapeinside={(*}{*)},
       alsoletter={?,:, ., !,, [,] \, ,},	
	emphstyle={[1]\sf\fontseries{b}\selectfont},
	sensitive=true,
	aboveskip=\smallskipamount, 
	belowskip=\smallskipamount, 
	morecomment=[l][\color{gray_ulisses}\sf\fontseries{m}\selectfont\itshape\small]{\#},
	stringstyle=\color{rubystr},
	showstringspaces=false,
	showspaces=false,
	breaklines=true,
	showtabs=false,
}

\lstdefinelanguage{sql}{
  basicstyle=\sf\small\fontseries{m}\selectfont,
        columns=flexible,
  emph={[1]
      SELECT, WHERE, ON, IN, INNER, JOIN, FROM},
  emphstyle={[1]\sf\fontseries{b}\selectfont},
  sensitive=true,
  aboveskip=\smallskipamount, 
  belowskip=\smallskipamount, 
  morecomment=[l][\color{gray_ulisses}\sf\fontseries{m}\selectfont\itshape\codesize]{\#\#}
}

\lstnewenvironment{rcode}
{\lstset{language=ruby,breaklines=true}
}
{}

\lstnewenvironment{sqlcode}
{\lstset{language=sql,breaklines=true}
}
{}

\lstnewenvironment{rcodebox}
{\lstset{language=ruby,breaklines=true,numbers=left,stepnumber=1,firstnumber=1,numberfirstline=true, frame=single, numberstyle={\footnotesize\it\color{rubystr}}}
}
{}

\newcommand{\code}[1]{\lstinline[language=ruby,mathescape,basicstyle=\sf\fontseries{m}\selectfont\normalsize,breaklines=true,breakatwhitespace=false,xleftmargin=0pt,xrightmargin=0pt]~#1~}

\def\spmid{\ \mid \ }
\def\MK#1{\textcolor{blue}{\sf MK:$\clubsuit$ #1$\clubsuit$}}

\def\NV#1{\textcolor{magenta}{\sf NV:$\heartsuit$ #1$\heartsuit$}}
\def\emphbf#1{\textbf{\emph{#1}}}

\newcommand{\class}[1]{\textsf{#1}}

\newcommand{\unrefine}[1]{\ensuremath{\llfloor {#1} \rrfloor}\xspace}
\newcommand{\judgementHead}[2]{\ensuremath{\emphbf{#1}\hfill\fbox{#2}}}
\newcommand{\teval}[4]{\ensuremath{#2 \vdash_{#1} #3:#4}}
\newcommand{\tcheck}[5]{\ensuremath{#2 \vdash_{#1} #3\hookrightarrow #4:#5}}
\newcommand{\progtyp}[2]{\ensuremath{#1 \models #2}\xspace}
\newcommand{\bturnstile}[1]{\ensuremath{\vdash_{#1, Base}}}
\newcommand{\bteval}[4]{\ensuremath{#2 \bturnstile{#1} #3:#4}}
\newcommand{\cteval}[3]{\ensuremath{#1(#2) = #3}}
\newcommand{\eeval}[2]{\ensuremath{#1 \rightsquigarrow #2}}
\newcommand{\eevalmulti}[2]{\ensuremath{#1 \rightsquigarrow^* #2}}
\newcommand{\typegoesto}[3]{\ensuremath{\langle #1, #2 \rangle \Downarrow #3}\xspace}
\newcommand{\subst}[2]{\ensuremath{[ #1 \mapsto #2 ]}}
\newcommand{\envbind}[2]{\ensuremath{#1 \text{:} #2}}

\newcommand{\trec}{\texttt{tself}\xspace}
\newcommand\trspace{\vspace{0.8em}}
\newcommand{\subtype}[2]{\ensuremath{#1\leq #2}}
\newcommand\tuniv{\textit{Type}\xspace}
\newcommand{\defof}[3]{\ensuremath{\textit{def\_of}(#1) = #2.#3}}
\newcommand{\dyn}[3]{\ensuremath{\langle #1, #2, #3 \rangle}\xspace}
\newcommand{\typeof}{\textit{type\_of}}
\newcommand{\classtable}{\textit{CT}\xspace}
\newcommand{\typestack}{\textit{TS}}
\newcommand{\stacktop}[2]{\ensuremath{#1 :: #2}}
\newcommand{\tstackelement}[3]{\ensuremath{(#1[#2], #3)}}
\newcommand{\dstackelement}[2]{\ensuremath{(#1, #2)}}
\newcommand{\consistent}[2]{\ensuremath{#1\ \sim\ #2}}
\newcommand{\builtin}{\ensuremath{\mathcal{L}}\xspace}
\newcommand{\userdef}{\mathcal{U}\xspace}
\newcommand{\callop}{\textit{call}\xspace}
\newcommand{\call}[3]{\ensuremath{\callop(#1, #2, #3)}}
\newcommand{\valid}{\ensuremath{\textit{valid}}}

\newcommand\ftypelabel{\texttt{ftype}\xspace}
\newcommand{\ftype}[2]{\ensuremath{\ftypelabel\ #1:#2}}
\newcommand{\rdl}{RDL\xspace}
\newcommand{\tocite}{(\textbf{cite})}

 \newcommand{\name}{CompRDL\xspace}
\newcommand{\comptype}{comp type\xspace} 
\newcommand{\comptypes}{comp types\xspace} 
\newcommand{\Comptype}{Comp type\xspace} 
\newcommand{\Comptypes}{Comp types\xspace} 
\newcommand{\CompType}{Comp Type\xspace}
\newcommand{\CompTypes}{Comp Types\xspace}
\newcommand{\benchtime}[2]{#1 \scriptsize{$\pm$ #2}}

\makeatletter
\newcommand\newsubcommand[3]{\newcommand#1{#2\sc@sub{#3}}}
\def\sc@sub#1{\def\sc@thesub{#1}\@ifnextchar_{\sc@mergesubs}{_{\sc@thesub}}}
\def\sc@mergesubs_#1{_{\sc@thesub#1}}
\makeatother

\newcommand\corelang{\ensuremath{\lambda^{C}}\xspace}
\newcommand\dynenv{\ensuremath{E}\xspace}

\newcommand\stack{\ensuremath{S}\xspace}
\newcommand\context{\ensuremath{C}\xspace}
\newcommand\val{\ensuremath{v}\xspace}
\newcommand\expr{\ensuremath{e}\xspace}
\newcommand\prog{\ensuremath{P}\xspace}

\newcommand\type{\ensuremath{A}\xspace}

\newcommand\depmethtype[4]{\ensuremath{(\typevar\text{<:}\mthtype{#1\slash #2)}{#3 \slash #4}}}
\newsubcommand\basetype{\type}{b}
\newcommand\methtype{\ensuremath{\sigma}\xspace}
\newcommand\dmethtype{\ensuremath{\delta}\xspace}
\newcommand\pdef[4]{\ensuremath{\texttt{def}\ #1(#2):#3\ =\ #4}\xspace}
\newcommand\bilabel{\texttt{lib}\xspace}
\newcommand\bidef[3]{\ensuremath{\bilabel\ #1(#2):#3}}

\newcommand\var{\ensuremath{x}\xspace}
\newcommand\typevar{\ensuremath{a}\xspace}
\newcommand\tnil{\textit{Nil}\xspace}
\newcommand\tobj{\textit{Obj}\xspace}
\newcommand\vnil{\texttt{nil}\xspace}
\newcommand\vtrue{\texttt{true}\xspace}
\newcommand\vfalse{\texttt{false}\xspace}
\newcommand\vinst[1]{\ensuremath{[#1]}\xspace}
\newcommand\eseq[2]{\ensuremath{#1 ; #2}\xspace}
\newcommand\elet[3]{\ensuremath{\texttt{let}\ #1 = #2\ \texttt{in}\ #3}\xspace}

\newcommand\fassn[2]{\ensuremath{#1 = #2}\xspace}
\newcommand\eif[3]{\ensuremath{\textbf{\texttt{if}}\ #1\ \textbf{\texttt{then}}\ #2\ \textbf{\texttt{else}}\ #3}\xspace}

\newcommand\eself{\texttt{self}\xspace}
\newcommand\emethcall[3]{\ensuremath{#1.#2(#3)}\xspace}
\newcommand\echeckmeth[4]{\ensuremath{\lceil #1 \rceil\emethcall{#2}{#3}{#4}}\xspace}
\newcommand\enew[1]{\ensuremath{#1 . \textbf{\texttt{new}}}\xspace}
\newcommand\mthtype[2]{\ensuremath{#1 \rightarrow #2}\xspace}
\newcommand\objrecord{\ensuremath{O}}
\newcommand\eq[2]{\ensuremath{#1==#2}}
\newcommand\tbool{\textit{Bool}\xspace}
\newcommand\ttrue{\textit{True}\xspace}
\newcommand\tfalse{\textit{False}\xspace}

\newcommand\rulename[1]{\textsc{#1}\xspace}
\newcommand\tappBase{\rulename{T-App}}

\newcommand\tappComp{\rulename{T-App-Lib}}
\newcommand\tNil{\rulename{T-Nil}}
\newcommand\tTrue{\rulename{T-True}}
\newcommand\tFalse{\rulename{T-False}}
\newcommand\tObject{\rulename{T-Object}}
\newcommand\tObj{\rulename{T-Obj}}
\newcommand\tType{\rulename{T-Type}}
\newcommand\tVar{\rulename{T-Var}}
\newcommand\tEq{\rulename{T-Eq}}
\newcommand\tLoc{\rulename{T-Loc}}
\newcommand\tSelf{\rulename{T-Self}}
\newcommand\tTSelf{\rulename{T-TSelf}}

\newcommand\tSeq{\rulename{T-Seq}}
\newcommand\tNew{\rulename{T-New}}
\newcommand\tIf{\rulename{T-If}}

\newcommand\cappBaseUD{\rulename{C-AppUD}}
\newcommand\cappBaseBI{\rulename{C-AppLib}}
\newcommand\cappComp{\rulename{C-App-Comp}}
\newcommand\cNil{\rulename{C-Nil}}
\newcommand\cTrue{\rulename{C-True}}
\newcommand\cFalse{\rulename{C-False}}

\newcommand\cType{\rulename{C-Type}}
\newcommand\cVar{\rulename{C-Var}}
\newcommand\cEq{\rulename{C-Eq}}

\newcommand\cObj{\rulename{C-Obj}}
\newcommand\cSelf{\rulename{C-Self}}
\newcommand\cTSelf{\rulename{C-TSelf}}

\newcommand\cSeq{\rulename{C-Seq}}
\newcommand\cNew{\rulename{C-New}}
\newcommand\cIf{\rulename{C-If}}

\newcommand\reVar{\rulename{E-Var}}
\newcommand\reSelf{\rulename{E-Self}}
\newcommand\reTSelf{\rulename{E-TSelf}}
\newcommand\reLet{\rulename{E-Let}}
\newcommand\reSeq{\rulename{E-Seq}}
\newcommand\reNew{\rulename{E-New}}
\newcommand\reLoc{\rulename{E-Loc}}
\newcommand\reEIfTrue{\rulename{E-IfTrue}}
\newcommand\reEIfFalse{\rulename{E-IfFalse}}
\newcommand\reEEqTrue{\rulename{E-EqTrue}}
\newcommand\reEEqFalse{\rulename{E-EqFalse}}

\newcommand\reAppUD{\rulename{E-AppUD}}
\newcommand\reContext{\rulename{E-Context}}
\newcommand\reAppBI{\rulename{E-AppLib}}
\newcommand\reRet{\rulename{E-Ret}}

\newcommand\tPDef{\rulename{T-PDef}}
\newcommand\tPBuiltIn{\rulename{T-PLib}}
\newcommand\tPSeq{\rulename{T-PSeq}}

\defcitealias{discourse}{Discourse 2018}

\begin{document}
\title[Type-Level Computations for Ruby Libraries]{Type-Level Computations for Ruby Libraries}


\author{Milod Kazerounian}
\affiliation{
  \institution{University of Maryland}
  \city{College Park}
  \state{Maryland}
  \country{USA}
}
\email{milod@cs.umd.edu}

\author{Sankha Narayan Guria}
\affiliation{
  \institution{University of Maryland}
  \city{College Park}
  \state{Maryland}
  \country{USA}
}
\email{sankha@cs.umd.edu}

\author{Niki Vazou}
\affiliation{
  \institution{IMDEA Software Institute}
  \city{Madrid}
  \country{Spain}
}
\email{niki.vazou@imdea.org}

\author{Jeffrey S. Foster}
\affiliation{
	\institution{Tufts University}
	\city{Medford}
	\state{Massachusetts}
	\postcode{02155}
	\country{USA}
}
\email{jfoster@cs.tufts.edu}

\author{David Van Horn}
\affiliation{
  \institution{University of Maryland}
  \city{College Park}
  \state{Maryland}
  \country{USA}
}
\email{dvanhorn@cs.umd.edu}

\begin{abstract}
  Many researchers have explored ways to bring static typing to
  dynamic languages. However, to date, such systems are not precise
  enough when types depend on values, which often arises
  when using certain Ruby libraries. For example, the type safety of
  a  database query in Ruby on Rails depends on the table and column
  names used in the query. To address this issue, we introduce
  CompRDL, a type system for Ruby that allows library method type
  signatures to include \emph{type-level computations} (or \emph{comp
    types} for short). Combined with singleton types for table and
  column names, comp types let us give database query methods type
  signatures that compute a table's schema to yield very precise type
  information. Comp types for hash, array, and string libraries can
  also increase precision and thereby reduce the need for type casts.
  We formalize CompRDL and prove its type system sound. Rather than
  type check the bodies of library methods with comp types---those
  methods may include native code or be complex---CompRDL inserts
  run-time checks to ensure library methods abide by their computed
  types. We evaluated CompRDL by writing annotations with type-level
  computations for several Ruby core libraries and database query
  APIs. We then used those annotations to type check two popular Ruby
  libraries and four Ruby on Rails web apps. We found the 
  annotations were relatively compact and could successfully
  type check 132 methods across our subject programs. Moreover, the use of
  type-level computations allowed us to check more expressive
  properties, with fewer manually inserted casts, than was possible
  without type-level computations.  In the process, we found two
  type errors and a documentation error that were confirmed by the developers. Thus, we believe
  CompRDL is an important step forward in bringing precise static type
  checking to dynamic languages.

\end{abstract}



\keywords{type-level computations, dynamic languages, types,
  Ruby, libraries, database queries}  

\maketitle
\renewcommand{\shortauthors}{M. Kazerounian, S.N. Guria, N. Vazou, D. Van Horn, J.S. Foster}

\section{Introduction}
\label{sec:introduction}

There is a large body of research on adding static typing to dynamic
languages~\cite{Furr:2009, Ren, ren:oops13, Tobin-Hochstadt2006,
  Tobin-Hochstadt2008, Anderson2005,
  Lerner2013,Thiemann2005,Ancona:2007,Aycock2000}. However, existing
 systems have limited support for the case when \emph{types}
depend on \emph{values}. Yet this case occurs surprisingly often,
especially in Ruby libraries.
For example, consider the following database query, written for
a hypothetical Ruby on Rails (a web framework, called Rails henceforth) app:
%
\begin{rcode}
Person.joins(:apartments).where({name: 'Alice', age: 30, apartments: {bedrooms: 2}})
\end{rcode}
This query uses the \emph{ActiveRecord} DSL to join two database
tables, \code{people}\footnote{Rails knows the plural of person is people.} and
\code{apartments}, and then filter on the
values of various columns (\code{name}, \code{age}, \code{bedrooms})
in the result.

We would like to type check such code, e.g., to ensure the columns
exist and the values being matched are of the right types. But we face
an important problem: what type signature do we give \textsf{joins}? Its
return type---which should describe the joined table---depends on the
value of its argument. Moreover, for $n$ tables, there are $n^2$ ways
to join two of them, $n^3$ ways to join three of them,
etc. Enumerating all these combinations is impractical.

To address this problem, in this paper we introduce \name{}, which
extends RDL~\cite{rdl-github}, a Ruby type system, to include method
types with \emph{type-level computations}, henceforth referred to as
\textit{\comptypes}. More specifically, in \name{} we can annotate
library methods with type signatures in which Ruby expressions can
appear as types. During type checking, those expressions are evaluated
to produce the actual type signature, and then typing proceeds as
usual. For example, for the call to \textsf{Person.joins}, by using a
singleton type for \code{:apartments}, a type-level computation can
look up the database schemas for the receiver and argument and then
construct an appropriate return type.\footnote{The use of
type-level computations and singleton types could be considered
dependent typing, but as our type system is much more restricted we introduce new
terminology to avoid confusion (see \S~\ref{subsec:discussion} for discussion).}

Moreover, the same type signature can work for any model class and any
combination of joins. And, because \name{} allows arbitrary
computation in types, \name{} type signatures have access to the full,
highly dynamic Ruby environment.  This allows us to provide very
precise types for the large set of Rails database query methods. It
also lets us give precise types to methods of \emph{finite hash types}
(heterogeneous hashes), \emph{tuple types} (heterogeneous arrays), and
\emph{const string types} (immutable strings), which can help
eliminate type casts that would otherwise be required. 

Note that in all these
cases, we apply \comptypes to library methods whose bodies we do not type check, in
part to avoid complex, potentially undecidable reasoning about whether a method body matches a
\comptype, but more practically because those library methods are
either implemented in native code (hashes, arrays, strings) or are
complex (database queries).  This design choice makes \name a particularly
practical system which we can apply to real-world programs. 
To maintain soundness, we insert dynamic
checks to ensure that these methods abide by their computed types at
runtime.  (\S~\ref{sec:overview} gives an overview of typing in
\name{}.)

We introduce \corelang{}, a core, object-oriented language that
formalizes \name{} type checking. In \corelang, library methods can be
declared with signatures of the form $\depmethtype{\expr_1}{A_1}{\expr_2}{A_2}$,
where $A_1$ and $A_2$ are the conventional (likely overapproximate)
argument and return types of the method. The precise argument and
return types are determined by evaluating $\expr_1$ and $\expr_2$,
respectively, and that evaluation may refer to the type of the
receiver and the type $a$ of the argument. \corelang also performs
type checking on $\expr_1$ and $\expr_2$, to ensure they
do not go wrong. To avoid potential infinite recursion, \corelang does
not use type-level computations during this type checking process, instead using
the conventional types for library methods. Finally, \corelang
includes a rewriting step to insert dynamic checks to ensure library
methods abide by their computed types. We prove  \corelang{}'s
type system is sound. (See \S~\ref{sec:formalism} for our formalism.)

We implemented \name on top of RDL, an existing Ruby type checker.
Since \name can include type-level computation that relies on mutable
values, \name inserts additional runtime checks to ensure such
computations evaluate to the same result at method call time as they
did at type checking time. Additionally, \name{} uses a lightweight
analysis to check that
type-level computations (and thus type checking) terminate.
The termination analysis uses purity effects to
check that calls that invoke iterator methods---the main source of
looping in Ruby, in our experience---do not mutate the receiver, which
could introduce non-termination.  Finally, we found that several kinds
of \comptypes we developed needed to include weak type updates to
handle mutation in Ruby programs.  (\S~\ref{sec:implementation}
describes our implementation in more detail.)

We evaluated \name{}
by first using it to write type annotations for 482 Ruby core library
methods and 104 Rails database query methods. We found that by using
helper methods, we could write very precise type annotations for all
586 methods with just a few lines of code on average. Then, we used
those annotations to type check 132 methods across two Ruby APIs and
four Ruby on Rails web apps.  We were able to successfully type check
all these methods in approximately 15 seconds total. In doing so, we
also found two type errors and a documentation error, which we confirmed with the developers.
We also found that, with \comptypes, type checking these benchmarks
required 4.75$\times$ fewer type cast annotations compared to
standard types, demonstrating \comptypes' increased
precision. (\S~\ref{sec:experiments} contains the results of our
evaluation.)



Our results suggest that using type-level
computations provides a powerful, practical, and precise way to
statically type check code written in dynamic languages.

\section{Overview}
\label{sec:overview}


The starting point for our work is RDL~\cite{rdl-github}, a system for
adding type checking and contracts to Ruby programs. RDL's type system
is notable because type checking statically analyzes source code, but
it does so \emph{at runtime}. For example, line~\ref{line:typesig} in
Figure~\ref{fig:discourse} gives a type signature for the method
defined on the subsequent line. This ``annotation'' is actually a call
to the method \code{type},\footnote{In Ruby, parentheses in a method
  call are optional.} which stores the type signature in a global
table. The type annotation includes a label \code{:model}. (In Ruby,
strings prefixed by colon are \emph{symbols}, which are interned
strings.)  When the program subsequently calls \code{RDL.do_typecheck
  :model} (not shown), RDL will type check the source code of all
methods whose type annotations are labeled \code{:model}.

This design enables RDL to support the metaprogramming that is common
in Ruby and ubiquitous in Rails. For example, the programmer can
perform type checking after metaprogramming code has run, when
corresponding type definitions are available. See \citet{Ren} for more
details.  We note that while \name benefits from this runtime type
checking approach---we use RDL's representation of types in our \name
signatures, and our subject programs include Rails apps---there is
nothing specific in the design of \comptypes that relies on it, and
one could implement \comptypes in a fully static system.

\subsection{Typing Ruby Database Queries}
\label{subsec:databases}


While RDL's type system is powerful enough to type check Rails apps in
general, it is actually very imprecise when reasoning about database
(DB) queries.
For example, consider Figure~\ref{fig:discourse}, which shows some code
from the \emph{Discourse} app. Among others, this app uses two tables,
\code{users} and \code{emails}, whose schemas are shown
on lines~\ref{line:users-schma} and~\ref{line:emails-schma}.
Each user has an \code{id}, a \code{username}, and a flag indicating
whether the account was \emph{staged}. Such staged accounts were
created automatically by \emph{Discourse} and can be claimed by the
email address owner. An email has an \code{id}, the email address, and
the \code{user_id} of the user who owns the email address.

\begin{figure}
\begin{subfigure}{\columnwidth}
  \begin{rcodebox}
# Table Schema
# users: { id: Integer, username: String, staged: bool } (*\label{line:users-schma}*)
# emails: { id: Integer, email: String, user_id: Integer } (*\label{line:emails-schma}*)

class User < ActiveRecord::Base
  type "(String, String) -> %bool", typecheck: :model (*\label{line:typesig}*)
  def self.available?(name, email)
      return false if reserved?(name) (*\label{line:reserved}*)
      return true if !User.exists?({username: name}) (*\label{line:userexists}*)
      # staged user accounts can be claimed
      return User.joins(:emails).exists?({staged: true, username: name, emails: { email: email }}) (*\label{line:userjoins}*)
   end
end
\end{rcodebox}
\caption{\label{fig:discourse} \textit{Discourse} code (uses \emph{ActiveRecord}).}
\end{subfigure}

\bigskip{}

\begin{subfigure}{\columnwidth}
\begin{rcodebox}
type Table, :exists?, "(<<schema_type(tself)>>) -> Boolean" (*\label{line:thash}*)
type Table, :joins, "(t<:Symbol) -> (*\label{line:tjoins}*)
	<<if t.is_a?(Singleton) 
	then Generic.new(Table, schema_type(tself).merge(  {t.val=>schema_type(t)}))
	else Nominal.new(Table) 
	end >>"

def schema_type(t) (*\label{line:schematype}*)
   if t.is_a?(Generic) && (t.base == Table)  # Table<T>(*\label{line:gentab}*)
      return t.param # return T
   elsif t.is_a?(Singleton)   # Class or :symbol(*\label{line:schmsym}*)
      table_name = t.val # get the class/symbol vale
      table_type = RDL.db_schema[table_name]
      return table_type.param      
   else # will only be reached for the nominal type Table
     return ... # returns Hash<Symbol, Object>
   end
end
\end{rcodebox}
\caption{\Comptype annotations for  query methods.}
\label{fig:comptype}
\end{subfigure}

\caption{Type Checking Database Queries in \textit{Discourse}.}
\label{fig:db-example}
\end{figure}

Next in the figure, we show code for the class \code{User}, which is a
\emph{model}, i.e., instances of the class correspond to rows in the
\code{users} table.  This class has one method, \textsf{available?},
which returns a boolean indicating whether the username and email
address passed as arguments are available. The method first checks
whether the username was already reserved (line~\ref{line:reserved},
note the postfix \code{if}). If not, it uses the database query method
\textsf{exists?} to see if the username was already taken
(line~\ref{line:userexists}). (Note that in Ruby,
\code{\{a: b\}} is a hash that maps the symbol \code{:a}, which is
suffixed with a colon when used as a key, to the value \code{b}.)
Otherwise, line~\ref{line:userjoins} uses a more complex query to
check whether an account was staged. More specifically, this code
joins the \code{users} and \code{emails} table and then looks for a
match across the joined tables.

We would like to type check the \textsf{exists?} calls in this code to
ensure they are type correct, meaning that the columns they refer to
exist and the values being matched are of the right type. The call on
line~\ref{line:userexists} is easy to check, as RDL can type the
receiver \code{User} as having an \code{exists?} method that takes a
particular \emph{finite hash type} \code{\{c1: t1, ..., cn: tn\}} as
an argument, where the \code{ci} are \emph{singleton types} for
symbols naming the columns, and the \code{ti} are the corresponding
column types.

Unfortunately, the \textsf{exists?} call on line~\ref{line:userjoins}
is another story. Notice that this query calls \textsf{exists?}  on
the result of \textsf{User.\-joins(:emails)}. Thus, to give
\textsf{exists?} a type with the right column information, we need to
have that information reflected in the return type of \textsf{joins}.
Unfortunately, there is no reasonable way to do this in RDL, because
the set of columns in the table returned by \textsf{joins} depends on both the
receiver and the \emph{value} of the argument. We could in theory
overload \textsf{joins} with different return types depending on the
argument type---e.g., we could say that \textsf{User.joins} returns a
certain type when the argument has singleton type
\textsf{:emails}. However, we would need to generate such signatures
for every possible way of joining two tables together, three tables
together, etc., which quickly blows up. Thus, currently, RDL types
this particular \textsf{exists?} call as taking a \code{Hash<Symbol,
  Object>}, which would allow type-incorrect arguments.


\paragraph*{\Comptypes for DB Queries.}

To address this problem, \name allows method type signatures to
include computations that can, on-the-fly, determine the method's
type. Figure~\ref{fig:comptype} gives \comptype signatures for
\textsf{exists?} and \textsf{joins}. It also shows the definition
of a helper method, \code{schema_type}, that is called from the \comptypes.
The \comptypes also make use of a new generic type
\code{Table<T>} to type a DB table whose columns are described by
\code{T}, which should be a finite hash type.

Line~\ref{line:thash} gives the type of \textsf{exists?}. Its argument
is a \comptype, which is a Ruby expression, delimited by
\guillemotleft$\cdot$\guillemotright, that evaluates to a standard
type. When type checking a call to \textsf{exists?} (including those
in the body of \textsf{available?}), CompRDL runs the \comptype
code to yield a standard type, and then proceeds with type
checking as usual with that type.

In this case, to compute the argument type for \textsf{exists?}, we
call the helper method \textsf{schema\_type} with \trec, which is a reserved variable naming
the type of the receiver. The \textsf{schema\_type} method has a few
different behaviors depending on its argument. When given a type
\code{Table<T>}, it returns \code{T}, i.e., the finite hash type
describing the columns. When given a singleton type representing a
class or a symbol, it uses another helper method \code{RDL.db_schema}
(not shown) to look up the corresponding
table's schema and return an appropriate finite hash type. Given any other
type, \textsf{schema\_type} falls back to returning the type
\code{Hash<Symbol, Object>}.

This type signature already allows us to type check the
\textsf{exists?} call on line~\ref{line:userexists}. On this line, the
receiver has the singleton type for the \code{User} class, so
\textsf{schema\_type} will use the second arm of the conditional and
look up the schema for \textsf{User} in the DB.

Line~\ref{line:tjoins} shows the \comptype signature for
\textsf{joins}. The signature's input type binds \code{t} to the actual argument type, and
requires it to be a subtype of \code{Symbol}. For example, for the call
on line~\ref{line:userjoins}, \textsf{t} will be bound to the singleton
type for \textsf{:emails}. The return \comptype can then refer to
\code{t}.  Here, if \textsf{t} is a singleton type, \textsf{joins}
returns a new \code{Table} type that merges the
schemas of the receiver and the argument tables using
\code{schema\_type}. Otherwise, it falls back to producing a
\code{Table} with no schema information. Thus,
the \code{joins} call on line~\ref{line:userjoins} returns type
\begin{center}
	\textsf{\small Table<\{staged:\%bool, username:String, id: Integer,\\
		emails: \{email:String, user\_id: Integer \}\}>}
\end{center}

That is, the type reflects the schemas of both the \code{users}
and \code{emails} tables.
Given this type, we can now type check the \textsf{exists?} call on
line~\ref{line:userjoins} precisely. On this line, the receiver has the
table type given above, so when called by \textsf{exists?} the helper
\textsf{schema\_type} will use the first arm of the conditional and
return the \code{Table} column types, ensuring the query is type checked
precisely.

Though we have only shown types for two query methods
in the figure, we note that \comptypes are easily extensible to other 
kinds of queries. Indeed, we have applied them to 104 methods
across two DB query frameworks (\S~\ref{sec:experiments}).
Furthermore, 
we can also use \comptypes to encode sophisticated invariants.
For example, in Rails, database tables can only be joined if
the corresponding classes have a declared \emph{association}.
We can write a \comptype for \code{joins} that enforces this.
(We omitted this in Figure~\ref{fig:db-example}
for brevity.)

Finally, we note that while we include a ``fallback'' case that
allows \comptypes to default to less precise types when necessary,
in practice this is rarely necessary for DB queries. That is,
parameters that are important for type checking, such as the name of tables
being queried or joined, or the names of columns be queried, 
are almost always provided statically in the code.

\subsection{Avoiding Casts using \CompTypes}
\label{subsec:overview:hashes}

\begin{figure}
\begin{rcodebox}
type Hash, :[], "(k) -> v" 
type Array, :first, "() -> a" 
type :page, "() -> {info: Array<String>, title: String}''

type "() -> String"
def image_url()
  page[:info].first  # can't type check
  # Fix: RDL.type_cast(page[:info], "Array<String>").first (*\label{line:typecast}*)
 end	
\end{rcodebox}

\caption{Type Casts in a Method.}
\label{fig:wiki}
\end{figure}

In addition to letting us find type errors in code we could not
previously type check precisely enough, the increased precision
of \comptypes can also help eliminate type casts.

For example, consider the code in Figure~\ref{fig:wiki}. The first
line gives the type signature for a method of \code{Hash}, which is
parameterized by a key type \code{k} and a value type \code{v}
(declarations of the parameters not shown). The specific method is
\textsf{Hash\#[]},\footnote{Here we use the Ruby idiom that
  \textsf{A\#m} refers to the instance method \textsf{m} of class
  \textsf{A}.} which, given a key, returns the corresponding
value. Notably, the form \code{x[k]} is desugared to \code{x.[](k)},
and thus hash lookup, array index, and so forth are methods rather
than built-in language constructs.

The second line similarly gives a type for \textsf{Array\#first},
which returns the first element of the array. Here type variable
\code{a} is the array's contents type (declaration also not
shown). The third line gives a type for a method \code{page} of the
current class, which takes no arguments and returns a hash in which
\code{:info} is mapped to an \code{Array<String>} and \textsf{:title} is mapped
to a \code{String}.

Now consider type checking the \code{image_url} method defined at the
bottom of the figure. This code is extracted and simplified from a
\emph{Wikipedia} client library used in our experiments (\S~\ref{sec:experiments}). Here, since
\code{page} is a no-argument method, it can be invoked without any
parentheses. We then invoke \textsf{Hash\#[]} on the result.

Unfortunately, at this point type checking loses precision. The
problem is that whenever a method is invoked on a finite hash type
\code{\{c1: t1, ..., cn: tn\}}, RDL (retroactively) gives up tracking
the type precisely and promotes it to \code{Hash<Symbol, t1 or...or
  tn>}~\cite{rdl-github}. In this case, \code{page}'s return type is
promoted to \code{Hash<Symbol, Array<String> or String>}.

Now the type checker gets stuck. It reasons that \code{first} could be
invoked on an array or a string, but \code{first} is defined only for
the former and not the latter. The only currently available fix is to
insert a type cast, as shown in the comment on line~\ref{line:typecast}.

One possible solution would be to add special-case support for
\code{[]} on finite hash types. However, this is only one of 54
methods of \code{Hash},
which is a lot of behavior to special-case. Moreover, Ruby programs can
\emph{monkey patch} any class, including \code{Hash},
to change library methods' behaviors. This makes building special support
for those methods inelegant and potentially brittle since the
programmer would have no way to adjust the typing of those methods.

In \name, we can solve this problem with a \comptype annotation. More
specifically, we can give \textsf{Hash\#[]} the following type:
\begin{rcode}
  type Hash, :[], "(t<:Object) -> 
    <<if tself.is_a?(FiniteHash) && t.is_a?(Singleton)
        then tself.elts[t.val]
        else tself.value_type end>>"
\end{rcode}
This \comptype specifies that if the receiver has a finite hash type
and the key has a singleton type, then \textsf{Hash\#[]} returns the
type corresponding to the key, otherwise it returns a value type
covering all possible values (computed by \code{value_type},
definition not shown).

Notice that this signature allows \code{image_url} to type check
without any additional casts. The same idea can be applied to many
other \code{Hash} methods to give them more precise types.

\paragraph*{Tuple Types.}

In addition to finite hash types, RDL has a special \emph{tuple type}
to model heterogeneous \code{Array}s. As with finite hash types, RDL
does not special-case the \code{Array} methods for tuples, since there
are 124 of them. This leads to a loss of precision when invoking
methods on values with tuple types. However, analogously to finite
hash tables, \comptypes can be used to recover precision. As examples,
the \textsf{Array\#first} method can be given a \comptype which returns 
the type of the first element of a tuple, and
the \comptype for \textsf{Array\#[]} has essentially the same logic as
\textsf{Hash\#[]}.

\paragraph*{Const String Types.}

As another example, Ruby strings are mutable, hence RDL does not give
them singleton types. (In contrast, Ruby symbols are immutable.) This
is problematic, because types might depend on string values. In
particular, in the next section we explore reasoning about string
values during type checking raw SQL queries.

Using \comptypes, we can assign singleton types to strings wherever possible.
We introduce a new \emph{const string} type representing strings that
are never written to. \name treats const strings as singletons, and
methods on \code{String} are given \comptypes that perform precise
operations on const strings and fall back to the \code{String} type as
needed. We discuss handling mutation for const strings, finite hashes, and tuples in
Section~\ref{sec:implementation}.

\subsection{SQL Type Checking}
\label{subsec:sqltypecheck}

\begin{figure}
\begin{rcodebox}
# Table Schema
# posts table { id: Integer, topic_id: Integer, ... }
# topics table { id: Integer, title: String, ... }
# topic_allowed_groups table { group_id: Integer, topic_id: Integer }

# Query with SQL strings
Post.includes(:topic)(*\label{line:sqlstr}*)
  .where('topics.title IN (SELECT topic_id FROM topic_allowed_groups WHERE `group_id` = ?)', self.id)

type Table, :where, "(t <: <<if t.is_a?(ConstString)(*\label{line:wheretype}*)
  then sql_typecheck(tself, t)
  else schema_type(tself)
  end >>) -> <<tself>>"
\end{rcodebox}

\caption{Type Checking SQL Strings in \textit{Discourse}.}
\label{fig:sql-example}
\end{figure}

As we saw in Figure~\ref{fig:db-example}, \emph{ActiveRecord} uses a
DSL that makes it easier to construct queries inside of Ruby. However,
sometimes programmers need to include raw SQL in their queries, either
to access a feature not supported by the DSL or to improve performance
compared to the DSL-generated query.

Figure~\ref{fig:sql-example} gives one such example, extracted and
simplified from \emph{Discourse}, one of our subject programs. Here
there are three relevant tables: \code{posts}, which stores posted
messages; \code{topics}, which stores the topics of posts; and
\textsf{topic\-\_allowed\-\_groups}, which is used to limit the topics allowed
by certain user groups.

Line~\ref{line:sqlstr} shows a query that includes raw SQL. First, the
\code{posts} and \code{topics} tables are joined via the
\code{includes} method. (This method does eager loading whereas
\code{joins} does lazy loading.) Then \code{where} filters the
resulting table based on some conditions. In this case, the conditions
involve a nested SQL query, which cannot be expressed except using raw
SQL that will be inserted into the final generated query.

This example also shows another feature: any \code{?}'s that appear in
raw SQL are replaced by additional arguments to \code{where}. In this
case, the \code{?} will be replaced by \textsf{self.id}.

We would like to extend type checking to also reason about the raw SQL
strings in queries, since they may have errors. In this particular
example, we have injected a bug. The inner \code{SELECT} returns a set
of integers, but \code{topics.title} is a string, and it is a type
error to search for a string in an integer set.

To find this bug, we developed a simple type checker for a subset of
SQL, and we wrote a \comptype for \textsf{where} that invokes it as
shown on line~\ref{line:wheretype}. In particular, if the type of the argument to
\code{where}, here referred to by \code{t}, 
is a const string,
then we type check that string as raw SQL, and
otherwise we compute the valid parameters of \code{where} using the
\code{schema_type} method from Figure~\ref{fig:db-example}. The result
of \code{where} has the same type as the receiver.

The \code{sql_typecheck} method (not shown) takes the receiver type,
which will be a \class{Table} with a type parameter describing the
schema, and the SQL string. 
One challenge that arises in type checking the SQL string is that it is actually
only a fragment of a query, which therefore cannot be directly parsed using
a standard SQL parser. We solve this problem by creating a complete,
but artificial, SQL query into which we inject the fragment. This query
is never run, but it is syntactically correct so it can be parsed.
Then, we replace any
\code{?}'s with placeholder AST nodes that store the types of the
corresponding arguments.

For example, the raw SQL in Figure~\ref{fig:sql-example} gets
translated to the following SQL query:
\begin{sqlcode}
SELECT * FROM `posts` INNER JOIN `topics`
  ON a.id = b.a_id
  WHERE topics.title IN (SELECT topic_id FROM topic_allowed_groups WHERE `group_id` = [Integer])
\end{sqlcode}
Notice the table names (\code{posts}, \code{topics}) occur on the
first line and the \code{?} has been replaced by a placeholder
indicating the type \code{Integer} of the argument. Also note that the
column names to join on (which are arbitrary here) are ignored by our
type checker, which currently only looks for errors in the where
clause.

Once we have a query that can be parsed, we can type check it using
the DB schema. In this case, the type mismatch between
\textsf{topics.title} and the inner query will be reported.

In \S~\ref{subsec:databases}, \comptypes were evaluated
to produce a normal type signature.
However, we use \comptypes in a slightly different way for checking
SQL strings.
The \code{sql_typecheck} method will itself perform type checking
and provide a detailed message when an error is found.
If no error is found, \code{sql_typecheck} will simply return the
type \textsf{String}, allowing type checking to proceed.

\subsection{Discussion}
\label{subsec:discussion}

Now that we have seen \name in some detail, we can discuss several
parts of its design.

\paragraph{Dynamic Checks.}

In type systems with type-level computations, or more generally
dependent type systems, comparing two types for equality is often
undecidable, since it requires checking if computations are equivalent.

To avoid this problem, \name only uses \comptypes for methods which
themselves are not type checked. For example, \textsf{Hash\#[]} is
implemented in native code, and we have not attempted to type check
\emph{ActiveRecord}'s \code{joins} method, which is part of a very
complex system.

As a result, type checking in \name is decidable. \Comptypes are only
used to type check method calls, meaning we will always have access
to the types of the receiver and arguments in a method call.
Additionally, in all cases we have encountered in practice, the types 
of the receiver and arguments are ground types (meaning they do not
contain type variables). 
Thus, \comptypes can be fully evaluated to non-\comptypes
before proceeding to type checking.

For soundness, since we do not type check the bodies of
\comptype-annotated methods, \name inserts dynamic checks at calls to
such methods to ensure they match their computed types. For example,
in Figure~\ref{fig:wiki}, \name inserts a check that
\code{page[:info]} returns an \code{Array}.
This follows the approach of gradual~\cite{Siek2006} and hybrid~\cite{Flanagan2006} typing,
in which dynamic checks guard 
statically unchecked code.


We should also note that although our focus is on applying \comptypes
to libraries, they can be applied to any method at the cost of dynamic
checks for that method rather than static checks. For example, they
could be applied to a user-defined library wrapper.

\paragraph{Termination.}

%
%

A second issue for the decidability of \comptypes is that type-level
computations could potentially not terminate. To avoid this
possibility, we implement a termination checker for \comptypes.  At a
high level, CompRDL ensures termination by checking that iterators used
by type-level code do not mutate their receivers and by
forbidding type-level code from using looping
constructs. We also assume there are no recursive method calls in
type-level code.  We discuss termination checking in more detail in
\S~\ref{sec:implementation}.

\paragraph{Value Dependency.} 

We note that, unlike dependent types
(e.g., Coq~\cite{Pierce:SFold}, Agda~\cite{Norell2009}, F*~\cite{Swamy:2016})
where types depend directly on terms, 
in \name types depend on the \textit{ types} of terms. 
For instance, in a \comptype \code{(t<:Object) -> tres}
the result type \code{tres} can depend on the type \code{t} of the argument.
Yet, since singleton types lift expressions into types, 
we could still use \name to express some value dependencies in types 
in the style of dependent typing. 


\paragraph{Constant Folding.}

Finally, in RDL, integers and floats have singleton types. Thus, we
can use \comptypes to lift some arithmetic computations to the type
level. For example, \name can assign the expression \code{1+1} the
type \textsf{Singleton(2)} instead of \code{Integer}. This effectively
incorporates constant folding into the type checker.

While we did write such \comptypes for \code{Integer} and \code{Float}
(see Table~\ref{tab:deptype-defn}), we found that this precision was
not useful, at least in our subject programs. The reason is that RDL
only assigns singleton types to constants, and typically arithmetic
methods are not applied to constant values. Thus, though we have
written \comptypes for the \code{Integer} and \code{Float} libraries,
we have yet to find a useful application for them in practice. We
leave further exploration of this topic to future work.

\section{Soundness of Comp Types}\label{sec:formalism}

In this section we formalize \name as \corelang, a core
object-oriented calculus that includes \comptypes for library
methods.  We first define the syntax and semantics of
\corelang (\S~\ref{subsec:formalism:syntax}), and then we formalize
type checking (\S~\ref{subsec:type-checking}). The type checking
process includes a rewriting step to insert dynamic checks to
ensure library methods satisfy their type signatures.
Finally, we prove type soundness (\S~\ref{subsec:formalism:properties}).
For brevity, we leave the full formalism and proofs to Appendix~\ref{sec:appendix}. Here we provide only the key details.

\subsection{Syntax and Semantics}\label{subsec:formalism:syntax}
\begin{figure}[t!]
	$$
	\begin{array}{lccl}
	\textit{Values} \quad
	& \val
	& ::= & \vnil \spmid \vtrue \spmid \vfalse \spmid A 
	\\
	\textit{Expressions} \quad
	& \expr
	& ::= & \val \spmid \var \spmid \typevar \spmid \eself \spmid \trec \\
	&     & \spmid & \enew{A} \spmid \eseq{\expr}{\expr} \spmid \eq{\expr}{\expr} \\
	&     & \spmid & \eif{\expr}{\expr}{\expr} \spmid \emethcall{\expr}{m}{\expr} \\ 
	&     & \spmid &  \echeckmeth{A}{\expr}{m}{\expr} 
	\\
	\textit{Meth. Types} \quad
	& \methtype
	& ::= & \mthtype{A}{A} \\
	\textit{Lib. Meth. Types} \quad
	& \dmethtype
	& ::= & \methtype \spmid \depmethtype{\expr}{A}{\expr}{A}
	\\
	\textit{Programs} \quad
	& \prog
	& ::= & \pdef{A.m}{\var}{\methtype}{\expr} \\
	&      & \spmid & \bidef{A.m}{\var}{\dmethtype} \spmid \eseq{P}{P}
	\\ 
	\textit{Type Env.} \quad
	& \Gamma
	& ::= & \emptyset \spmid \envbind{x}{A}\\
	\textit{Dyn. Env.} \quad
	& E
	& ::= & \emptyset \spmid \envbind{x}{\val} \\
	\textit{Class Table} \quad
	& \classtable
	& ::= & \emptyset \spmid \envbind{A.m}{\dmethtype},\classtable\\ 
	\textit{Method Sets} \quad
	& \userdef
	& : & \textrm{user-defined methods} \\
	& \builtin
	& : & \textrm{library methods} \\
	\end{array}
	$$
	
	$\var,\typevar\in\textrm{var IDs}$, 
	$m\in\textrm{method IDs}$, 
	$A\in\textrm{class IDs}$,
	$\userdef \cap \builtin = \emptyset$
	\caption{Syntax and Relations of \corelang.}
	\label{figure:lang}
\end{figure}	

\begin{figure*}
	\centering
	\judgementHead{Type Checking and Rewriting Rules}{\tcheck{\classtable}{\Gamma}{\expr}{\expr}{A}}
$$
\begin{array}{c}
  	\inference{}{\tcheck{\classtable}{\Gamma}{A}{A}{\tuniv}}[\cType]

  \qquad
  
	\inference{
		\tcheck{\classtable}{\Gamma}{\expr}{\expr'}{A} &&
		\cteval{\classtable}{A.m}{\mthtype{A_1}{A_2}} &&
		A.m\in\userdef \\
		\tcheck{\classtable}{\Gamma}{\expr_x}{\expr_x'}{A_x} &&
		\subtype{A_x}{A_1} \\
	}
	{\tcheck{\classtable}{\Gamma}{\emethcall{\expr}{m}{\expr_x}}{\emethcall{\expr'}{m}{\expr_x'}}{A_2}}[\cappBaseUD]
\\ \\
		\inference{
		\tcheck{\classtable}{\Gamma}{\expr}{\expr'}{A} \\
		\cteval{\classtable}{A.m}{\mthtype{A_1}{A_2}} \\
		A.m\in\builtin &&
		\subtype{A_x}{A_1} \\
		\tcheck{\classtable}{\Gamma}{\expr_x}{\expr_x'}{A_x}
	}
	{\tcheck{\classtable}{\Gamma}{\emethcall{\expr}{m}{\expr_x}}{\echeckmeth{A_2}{\expr'}{m}{\expr_x'}}{A_2}}[\cappBaseBI]

  \qquad
  
	\inference{
		\tcheck{\classtable}{\Gamma}{\expr}{\expr'}{A} &&
		\cteval{\classtable}{A.m}\depmethtype{\expr_{t1}}{A_{t1}}{\expr_{t2}}{A_{t2}} \\
		A.m\in \builtin &&
		\tcheck{\classtable}{\Gamma}{\expr_x}{\expr_x'}{A_x}  \\
		\tcheck{\unrefine{\classtable}}{\envbind{\typevar}{\tuniv}, \envbind{\trec}{\tuniv}}{\expr_{t1}}{\expr_{t1}'}{\tuniv} \\
		\typegoesto{\subst{\typevar}{A_x}\subst{\trec}{A}}{\expr_{t1}'}{A_1} && 
		\subtype{A_x}{A_1}\\
                \tcheck{\unrefine{\classtable}}{\envbind{\typevar}{\tuniv},\envbind{\trec}{\tuniv}}{\expr_{t2}}{\expr_{t2}'}{\tuniv} \\
                \typegoesto{\subst{\typevar}{A_x}\subst{\trec}{A}}{\expr_{t2}'}{A_2} \\
	}
	{\tcheck{\classtable}{\Gamma}{\emethcall{\expr}{m}{\expr_x}}{\echeckmeth{A_2}{\expr'}{m}{\expr_x'}}{A_2}}[\cappComp]
\end{array}	
$$	
	\caption{A subset of the type checking and rewriting rules for \corelang.}
	\label{figure:check-insertion}
\end{figure*}

Figure~\ref{figure:lang} gives the syntax of
\corelang.
\textit{Values} \val include \vnil, \vtrue, and \vfalse. To support
\comptypes, class IDs $A$, which are the base types in
\corelang, are also values. We assume the set of class IDs includes
several built-in classes: \tnil, the class of \vnil; \tobj, which is
the root superclass; \ttrue and \tfalse, which are the classes of \vtrue
and \vfalse, respectively, as well as their superclass \tbool; and
\tuniv, the class of base types $A$.


\textit{Expressions} \expr include values \val and variables \var and
\typevar. By convention, we use the former in regular program
expressions and the latter in \comptypes.
The special variable \eself names the receiver of a method call,
and the special variable \trec names the \textit{type} of the receiver
in a \comptype.
New object instances are created with \enew{A}.
Expressions also include 
sequences \eseq{\expr}{\expr},
conditionals \eif{\expr}{\expr}{\expr}, and 
method calls \emethcall{\expr}{m}{\expr}, where, to simplify the
formalism, methods take one argument.
Finally, our type system translates calls to library methods into
\emph{checked method calls} \echeckmeth{A}{\expr}{m}{\expr}, which
checks at run-time that the value returned from the call has type $A$.
We assume this form does not appear in the surface syntax.


%
We assume the classes form a lattice with \tnil as the bottom and
\tobj as the top.
We write the least upper bound of $A_1$ and $A_2$ as $A_1 \sqcup A_2$.
For simplicity, we assume the lattice correctly models
the program's classes,
i.e., if $A \leq A'$, then $A$ is a subclass of $A'$ by the usual
definition.
Lastly, three of the built-in classes, \tnil, \ttrue, and \tfalse, are
\emph{singleton types}, i.e., they contain only the values \vnil,
\vtrue, and \vfalse, respectively. Extending \corelang with support
for more kinds of singleton types is straightforward.

%
\textit{Method Types} \methtype 
are of the form 
\mthtype{A'}{A} where 
$A'$ and $A$ are the domain and range types, respectively.
\textit{Library Method Types}
\dmethtype are either method types or 
have the form \depmethtype{\expr'}{A'}{\expr}{A}, 
where $\expr'$ and $\expr$ are expressions that evaluate to types and
that can refer to the variables \typevar and \trec.
The base types $A'$ and $A$ provide an upper bound 
on the respective expression types, 
i.e., for any \typevar, expressions $\expr'$ and $\expr$ should evaluate to subtypes of  
$A'$ and $A$, respectively. These upper bounds are used for
type checking \comptypes (\S~\ref{subsec:type-checking}).

Finally, \textit{programs} are sequences of method definitions and
library method declarations.

\paragraph*{Dynamic Semantics.}
%
The dynamic semantics of \corelang 
are the small-step
semantics of~\citet{Ren},
modified to throw blame (\S~\ref{subsec:formalism:properties}) when a checked method call fails. 
They use \textit{dynamic environments} $E$,
defined in Figure~\ref{figure:lang}, which map variables to values.
We define the relation
$
\typegoesto{\dynenv}{\expr}{\expr'} 
$,
meaning the expression $\expr$ evaluates to $\expr'$
under dynamic environment \dynenv. The full evaluation rules
use a stack as well~\ref{sec:appendix},
but we omit the stack here for simplicity.

\paragraph*{Example.} 
As an example \comptype in the formalism, consider type
checking the expression \emethcall{\vtrue}{\land}{\vtrue}, where the
$\land$ method returns the logical conjunction of the receiver and
argument. Standard type checking would assign this expression the type
$\tbool$. However, with \comptypes we can do better.

Recall that \vtrue and \vfalse are members of the singleton types
\ttrue and \tfalse. Thus, we can write a \comptype for the $\land$
method that yields a singleton return type when the arguments are
singletons, and \tbool in the fallback case:
$$
\begin{array}{l}
\bidef{\tbool.\land}{\var}{\depmethtype{\tbool}{\tbool}{(
	\\\quad\quad \eif{(\eq{\trec}{\ttrue}).\land(\eq{\typevar}{\ttrue})}{\ttrue
	\\\quad\quad}{
     \eif{(\eq{\trec}{\tfalse}).\lor(\eq{\typevar}{\tfalse})}
			{\tfalse\\\quad\quad}{\tbool)}}}
	{\tbool}}
\end{array}
$$
The first two lines of the condition handle the singleton cases, and
the last line is the fallback case.



\subsection{Type Checking and Rewriting}
\label{subsec:type-checking}

Figure~\ref{figure:check-insertion} gives a subset of the rules for
type checking \corelang and rewriting \corelang to insert dynamic checks at
library calls. The remaining rules, which are straightforward, can be
found in Appendix~\ref{sec:appendix}. These rules use two additional
definitions from Figure~\ref{figure:lang}. \emph{Type environments}
$\Gamma$ map variables to base types, and the \emph{class table}
\classtable maps methods to their type signatures. We omit the
construction of class tables, which is standard. We also use disjoint
sets $\userdef$ and
$\builtin$ to refer to the user-defined and library methods,
respectively.

The rules in Figure~\ref{figure:check-insertion} prove judgments of
the form \tcheck{\classtable}{\Gamma}{\expr}{\expr'}{A}, meaning
under type environment $\Gamma$ and class table \classtable, source
expression \expr is rewritten to target expression $\expr'$, which has
type $A$.





Rule~(\cType) is straightforward: any class ID $A$ that is used as a
value is rewritten to itself, and it has type \tuniv. We include this
rule to emphasize that types are values in \corelang.

Rule~(\cappBaseUD) finds the receiver type $A$, then looks up
$A.m$ in the class table. This rule only applies when $A.m$ is user-defined
and thus has a
(standard) method type $\mthtype{A_1}{A_2}$.
Then, as is standard, the rule checks that the argument's type $A_x$
is a subtype of $A_1$, and the type of the whole call is $A_2$. This
rule rewrites the subexpressions $e$ and $e_x$, but it does not itself
insert any new checks, since user-defined methods are statically
checked against their type signatures (rule not shown).

Rule~(\cappBaseBI) is similar to Rule~(\cappBaseUD), except it applies
when the callee is a library method. In this case, the rule inserts a
check to ensure that, at run-time, the library method abides by its
specified type.


Rule (\cappComp) is the crux of \corelang's type checking system. It
applies at a call to a library method $A.m$ that uses a type-level
computation, i.e., with a type signature
$\depmethtype{\expr_{t1}}{A_{t1}}{\expr_{t2}}{A_{t2}}$.  The rule
first type checks and rewrites $\expr_{t1}$ and $\expr_{t2}$ to ensure
they will evaluate to a type (i.e., have type \tuniv). These
expressions may refer to $\typevar$ and $\trec$, which themselves have
type $\tuniv$. The rule then \emph{evaluates} the rewritten
$\expr_{t1}$ and $\expr_{t2}$ using the dynamic semantics mentioned
above to yield types $A_1$ and $A_2$, respectively. Finally, the rule
ensures that the argument $e_x$ has a subtype of $A_1$; sets the
return type of the whole call to $A_2$; and inserts a dynamic check
that the call returns an $A_2$ at runtime. For instance,
the earlier example of the use of logical conjunction would be rewritten
to \echeckmeth{\ttrue}{\vtrue}{\land}{\vtrue}.


There is one additional subtlety in Rule (\cappComp). Recall the
example above that gives a type to $\tbool.\land$. Notice that the
type-level computation itself uses $\tbool.\land$. This could
potentially lead to infinite recursion, where calling $\tbool.\land$
requires checking that $\tbool.\land$ produces a type, which
requires recursively checking that $\tbool.\land$ produces a type
etc.

To avoid this problem, we introduce a function $\unrefine{\classtable}$ that
rewrites class table \classtable to drop all annotations with
type-level expressions. More precisely, any \comptype
$\depmethtype{\expr_1}{A_1}{\expr_2}{A_2}$ is rewritten to
$\mthtype{A_1}{A_2}$. Then type checking type-level computations, in
the fifth and eighth premise of (\cappComp), is done under the
rewritten class table.


Note that, while this prevents the type checking rules from infinitely
recursing, it does not prevent type-level expressions from themselves
diverging. In \corelang, we assume this does not happen, but in
our implementation, we include a simple termination checker that is
effective in practice (\S~\ref{sec:implementation}).

\subsection{Properties of \corelang.}\label{subsec:formalism:properties}
Finally, we prove type soundness for \corelang.
For brevity, we provide only the high-level description of the proof. 
The details can be found in Appendix~\ref{sec:appendix}.

\paragraph{Blame.}
The type system of \corelang does not prevent null-pointer errors,
i.e., \vnil has no methods yet we allow it to appear wherever any
other type of object is expected.
We encode such errors as \textit{blame}.
We also reduce to \textit{blame} when a dynamic check of the form \echeckmeth{A'}{A}{m}{\val}
fails.

\paragraph{Program Checking and \classtable.}

In the Appendix~\ref{sec:appendix} we provide type checking rules
not just for \corelang expressions but also for programs $P$.
These rules are where we actually check user-defined
methods against their types. We also define a notion
of \textit{validity} for a class table \classtable with respect to $P$, which enforces
that \classtable's types for methods and fields match the 
declared types in $P$, and that appropriate subtyping
relationships hold among subclasses. Given a well typed program $P$,
it is straightforward to construct a valid \classtable.

\paragraph{Type Checking Rules.}
In addition to the type checking and rewriting rules of
Figure~\ref{figure:check-insertion}, we define a separate judgment
\teval{\classtable}{\Gamma}{\expr}{A} that is identical to
\tcheck{\classtable}{\Gamma}{\expr}{\expr}{A} except it omits the
rewriting step, i.e., only performs type checking.

We can then prove soundness of the judgment
\teval{\classtable}{\Gamma}{\expr}{\type} using preservation and
progress, and finally prove soundness of the type checking and
rewriting rules as a corollary:

\begin{theorem}[Soundness]\label{theorem:soundness}
For any expressions \expr and \expr', type~$A$, class table \classtable, and program $P$
such that \classtable is valid with respect to $P$,
\textit{if} \tcheck{\classtable}{\emptyset}{\expr}{\expr'}{A} 
\textit{then} $\expr'$ either reduces to a value, reduces to blame, or does not terminate.
\end{theorem}


\section{Implementation}
\label{sec:implementation}

We implemented \name{} as an extension to
\rdl, a type checking system for Ruby~\cite{rdl-github,Ren,
  strickland:dls12, ren:oops13}. In total, \name{} comprises
approximately 1,170 lines of code added to \rdl.

RDL's design made it straightforward to add \comptypes. We extended
RDL so that, when type checking method calls, type-level computations
are first type checked to ensure they produce a value of type $\tuniv$
and then are executed to produce concrete types, which are then used
in subsequent type checking. \Comptypes use RDL's contract mechanism
to insert dynamic checks for \comptypes.

\paragraph{Heap Mutation.}

For simplicity, \corelang does not include a heap. By contrast,
\name allows arbitrary Ruby code to appear in \comptypes. This allows
great flexibility, but it means such code might depend on mutable
state that could change between type checking and the execution of a
method call.  For example, in Figure~\ref{fig:db-example}, type-level
code uses the global table \code{RDL.db_schema}. If, after
type checking the method \textsf{available?}, the program (pathologically)
changed the schema of \code{User} to drop the \code{username}
column, then \code{available?} would fail at runtime even
though it had type checked. The dynamic checks discussed in
\S~\ref{sec:overview} and \S~\ref{sec:formalism} are insufficient
to catch this issue, because they only check a method call
against the initial result of evaluating a \comptype;
they do not consider that the same \comptype might yield
a new result at runtime.

To address this issue, \name extends dynamic checks to ensure types remain
the same between type checking and execution.  If a method call is
type checked using a \comptype, then prior to that call
at runtime, \name will reevaluate that same \comptype on the same
inputs. If it evaluates to a different type, \name will raise an
exception to signal a potential type error. An alternative approach
would be to re-check the method under the new type.

Of course, the evaluation of a comp type may itself alter mutable
state.  Currently, CompRDL assumes that
comp type specifications are correct, including any mutable
computations they may perform. If a comp type does have any erroneous effects,
program execution could fail in an unpredictable manner. Other researchers have proposed
safeguards for this issue of effectful contracts by using guarded locations~\cite{Dimoulas2012}
or region based effect systems~\cite{Sekiyama:2017}.
We leave incorporating such safeguards for comp types as future work.
We note, however, that this issue did not arise in any comp types we
used in our experiments.

\paragraph*{Termination of \CompTypes.}

A standard property of type checkers is that they terminate. However,
because \comptypes allow arbitrary Ruby code, \name could potentially
lose this property. To address this issue, \name includes a lightweight
termination checker for \comptypes.

\begin{figure}
\begin{rcodebox}
type :m1, ..., terminates: :+
type :m2, ..., terminates: :+
type :m3, ..., terminates: :-

type Array, :map, ..., terminates: :blockdep  (*\label{line:map-type}*)
type Array, :push, ..., pure: :-

def m1()
  m2() # allowed: m2 terminates (*\label{line:m2-term}*)
  m3() # not allowed: m3 may not terminate (*\label{line:m3-nonterm}*) 
  while ... end # not allowed: looping (*\label{line:m1-loop} *)
  
  array = [1,2,3] # create new array
  array.map { |val| val+1 } # allowed (*\label{line:map-allowed}*)
  array.map { |val| array.push(4) } (*\label{line:map-rejected}*)
  # not allowed: iterator calls impure method push
end      


\end{rcodebox}

\caption{Termination Checking with \name.}
\label{fig:termination}
\end{figure}

Figure~\ref{fig:termination} illustrates the ideas behind termination
checking. In \name, methods can be annotated with \emph{termination
  effects} \code{:+}, for methods that always terminate (e.g., \code{m1} and
\code{m2}) and \code{:-} for methods that might diverge (e.g., \code{m3}).
\name allows terminating methods to call other terminating methods
(Line~\ref{line:m2-term}) but not potentially non-terminating methods
(Line~\ref{line:m3-nonterm}). Additionally, terminating methods may
not use loops (Line~\ref{line:m1-loop}). CompRDL assumes
that type-level code does not use recursion, and leave checking
of recursion to future work.


We believe it is reasonable to forbid the use of built-in loop
constructs, and to assume no recursion, because in practice most
iteration in Ruby occurs via methods that iterate over a structure.
For instance, \code{array.map \{~block~\}} returns a new array in
which the \code{block}, a \emph{code block} or lambda, has been
applied to each element of \code{array}. Since arrays are by
definition finite, this call terminates as long as \code{block}
terminates and does not mutate the \code{array}. A similar argument holds
other iterators of \code{Array}, \code{Hash}, etc.

Thus, \name{} checks termination of iterators as follows. Iterator
methods can be annotated with the special termination effect
\code{:blockdep} (Line~\ref{line:map-type}), indicating the method
terminates if its block terminates and is pure. \name{} also includes
\emph{purity effect} annotations indicating whether methods are pure
(\code{:+}) or impure (\code{:-}). A pure method may not write to any
instance variable, class variable, or global variable, or call an
impure method. \name{} determines that a \code{:blockdep} method
terminates as long as its block argument is pure, and otherwise it may
diverge.  Using this approach, \name will allow
Line~\ref{line:map-allowed} but reject reject
Line~\ref{line:map-rejected}.

\paragraph*{Type Mutations and Weak Updates}

Finally, to handle aliasing, our type annotations for \code{Array},
\code{Hash}, and \code{String} need to perform weak updates to type
information when tuple, finite hash, and const string types,
respectively, are mutated. For example, consider the following code:
\begin{rcode}
a = [1, 'foo']; if...then b = a else...end; a[0]='one'
\end{rcode}
Here (ignoring singleton types for simplicity), \code{a} initially has
the type \code{t = [Integer, String]}, where \code{t} is a Ruby
object, specifically an instance of RDL's \code{TupleType} class. At
the join point after the conditional, the type of \code{b} will be a union
of \code{t} and its previous type.

We could potentially forbid the assignment to \code{a[0]} because the
right-hand side does not have the type \code{Integer}. However, this is
likely too restrictive in practice.  Instead, we would like to mutate
\code{t} after the write. However, \code{b} shares this type.  Thus we
perform a \emph{weak update}: after the assignment we mutate \code{t}
to be \code{[Integer $\cup$ String, String]}, to handle the cases when
\code{a} may or may not have been assigned to \code{b}.

For soundness, we need to retroactively assume \code{t} was always this
type. Fortunately, for all tuple, finite hash, and const string types $\tau$, RDL
already records all asserted constraints $\tau' \leq \tau$ and
$\tau \leq \tau'$ to support promotion of tuples, finite hashes, and const strings to types
\code{Array}, \code{Hash}, and \code{String}, respectively~\cite{rdl-github}. We use this
same mechanism to replay previous constraints on these
types whenever they are mutated. For example, if previously we had a
constraint $\alpha \leq [\code{Integer}, \code{String}]$, and subsequently we
mutated the latter type to [\code{Integer}$\cup$\code{String}, \code{String}], we
would ``replay'' the original constraint as
$\alpha \leq [\code{Integer}\cup \code{String}, \code{String}]$.


	\section{Experiments}
\label{sec:experiments}

We evaluated \name{} by writing \comptype annotations for
a number of Ruby core and third party libraries
(\S~\ref{subsec:library-types}) and using these types to type check
real-world Ruby applications (\S~\ref{subsec:benchmarks}). We discuss
the results of type checking these benchmarks,
including the type errors we found in the process
(\S~\ref{subsec:results}). In all, we wrote 586
\comptype annotations for Ruby library methods, used them to type check
132 methods across six Ruby apps, found three bugs in the process,
and used significantly fewer manually inserted type casts than
are needed using RDL.

\subsection{Library Types}
\label{subsec:library-types}

  \begin{table}
	\centering
	\small
	\caption{Library methods with \comptype definitions.}
	\begin{tabular}{|r|r|r|r|}
		\hline
		\thead{Library} & \thead{\CompType\\Definitions} & \thead{Ruby\\LoC} & \thead{Helper\\Methods} \\
		\hline
		\multicolumn{4}{|l|}{\textit{Ruby Core Library}} \\
		\hline
		\class{Array} & 114 & 215 & 15 \\
		\class{Hash} & 48 & 247 & 15 \\
		\class{String} & 114 & 178 & 12 \\
		\class{Float}* & 98 & 12 & 1 \\
		\class{Integer}* & 108 & 12 & 1 \\
		\hline
		\multicolumn{4}{|l|}{\textit{Database DSL}} \\
		\hline
		\class{ActiveRecord} & 77 & 375 & 18 \\
		\class{Sequel} & 27 & 408 & 22 \\
		\hline
		Total & 586 & 1447 & 83 \\
		\hline
		\multicolumn{4}{l}{\footnotesize${}^*$Helper methods for Float and Integer are shared.} \\
	\end{tabular}
	\label{tab:deptype-defn}
\end{table}

Table~\ref{tab:deptype-defn} details the library type annotations we
wrote.

We chose to define \comptypes for these libraries due to
their popularity and because, as discussed in
\S~\ref{sec:overview}, they are amenable to precise typing with
\comptypes. These types were written based on the libraries'
documentation as well as manual testing to ensure type specifications
matched associated method semantics.

\begin{itemize}
\item \textit{Ruby core libraries:} These are libraries that are
  written in C and automatically loaded in all Ruby programs. We
  annotate the methods from the \class{Array}, \class{Hash}, \class{String},
  \class{Integer}, and \class{Float}
  classes.
\item \textit{ActiveRecord:} \class{ActiveRecord} is the most used
  object-relational model (ORM) DSL of the Ruby on Rails web framework.
  We wrote \comptypes for \class{ActiveRecord} database query
  methods.
\item \textit{Sequel:} \code{Sequel} is an alternative database ORM DSL.
It offers some more expressive queries than are available in ActiveRecord.
\end{itemize}

Table~\ref{tab:deptype-defn} lists the number of methods for
which we defined \comptypes in each library and the number
of Ruby lines of code (LoC) implementing the type computation
logic. The LoC count was calculated with \code{sloccount}
\cite{sloc} and does not include the line of the type annotation
itself.

In developing \comptypes for these libraries, we discovered that many
methods have the same type
checking logic. This helped us
write \comptypes for entire libraries using a few common helper
methods. 
In total, we wrote \comptype 
annotations for 586 methods across these libraries, comprising 1447
lines of type-level code and using 83 helper methods.
Once written, these \comptypes can be used to type check
as many of the libraries' clients as we would like, making the effort
of writing them potentially very worthwhile. 


\begin{table*}
	\centering
	\small
	\caption{Type checking results.}
	\begin{tabular}{|r|rrrr|r|rrr|r|}
		\hline
		\thead{Program} & \thead{Meths} & \thead{LoC} &\thead{Extra\\Annots.} & \thead{Casts} & \thead{Casts\\(RDL)} & \thead{Time (s)\\Median $\pm$ SIQR} & \thead{Test Time\\No Chk (s)}&  \thead{Test Time\\w/Chk. (s)}&\thead{Errs} \\
		\hline
		\multicolumn{10}{|l|}{\textit{API client libraries}} \\
		\hline
		\textit{Wikipedia} & 16 & 47 & 3 & 1 & 13 & \benchtime{0.06}{0.00} & \benchtime{6.3}{0.13} & \benchtime{6.32}{0.11} &0 \\ 
		\textit{Twitter} & 3 & 29 & 11 & 3 & 8 & \benchtime{0.02}{0.00} & \benchtime{0.07}{0.00} & \benchtime{0.08}{0.00} &0 \\ 
		\hline
		\multicolumn{10}{|l|}{\textit{Rails Applications}} \\
		\hline
		\textit{Discourse} & 36 & 261 & 32 & 13 & 22 & \benchtime{7.77}{0.39} & \benchtime{80.24}{0.63} & \benchtime{81.04}{0.34}&0 \\ 
		\textit{Huginn} & 7 & 54 & 6 & 3 & 6 & \benchtime{2.46}{0.29} & \benchtime{4.30}{0.21} & \benchtime{4.59}{0.48} &0 \\ 
          \textit{Code.org} & 49
                                                & 530 & 53 & 3 & 68 & \benchtime{0.49}{0.01} & \benchtime{2.49}{0.13}& \benchtime{2.74}{0.02} &1\\ 
		\textit{Journey} & 21 & 419 & 78 & 14 & 59 & \benchtime{4.12}{0.08} & \benchtime{4.52}{0.22}&  \benchtime{4.76}{0.24}& 2\\ 
		\hline
		Total & 132 & 1340 & 183 & 37 & 176 & \benchtime{14.93}{0.77} & \benchtime{97.93}{1.31} & \benchtime{99.53}{1.20} &3\\ 
		\hline
	\end{tabular}
	\label{tab:type check-bench}
\end{table*}

\subsection{Benchmarks}
\label{subsec:benchmarks}
We evaluated \name by type checking methods from two popular
Ruby libraries and four Rails web apps:

\begin{itemize}
\item \textit{Wikipedia Client} \cite{wikipedia} is a Ruby wrapper library for the
  Wikipedia API.
\item \textit{Twitter Gem} \cite{twitter} is a Ruby wrapper library for the Twitter
  API.
\item \textit{Discourse} \cite{discourse} is an open-source discussion platform built
  on Rails. It uses \class{ActiveRecord}.
\item \textit{Huginn} \cite{huginn} is a Rails app for setting up agents
  that monitor the web for events and perform automated
  tasks in response. It uses \class{ActiveRecord}.
\item \textit{Code.org} \cite{cdo} is a Ruby app that powers
  \url{code.org}, a site that encourages people,
  particularly students, to learn programming. It uses a
  combination of \class{ActiveRecord} and \class{Sequel}.
\item \textit{Journey} \cite{journey} is a web application that provides a graphical
  interface to create surveys and collect responses from
  participants. It uses a combination of \class{ActiveRecord} and \class{Sequel}.
\end{itemize}

We selected these benchmarks because they are popular,
well-maintained, and make extensive use of the libraries noted in
\S~\ref{subsec:library-types}. More specifically, the APIs
often work with hashes representing JSON objects
received over HTTP, and the Rails apps rely heavily on database
queries.

Since \name{} performs type checking, we must provide a type annotation
for any method we wish to type check. Our subject programs are very
large, and hence annotating all of the programs' methods is
infeasible. Instead, we focused
on methods for which \comptypes would
be most useful.

In \textit{Wikipedia}, we annotated the entire \code{Page} API. To
simplify type checking slightly, we changed the code to replace string
hash keys with symbols, since RDL's finite hash types do not currently
support string keys. In \textit{Twitter}, we annotated all the methods
of stream API bindings that made use of methods with
\comptypes.

In \textit{Discourse} and \textit{Huginn}, we chose several larger
Rails model classes, such as a \code{User} class that represents
database rows storing user information. In \textit{Code.org} and
\textit{Journey}, we type checked all methods that used \class{Sequel}
to query the database.  Within the selected classes for these four
Rails apps, we annotated a subset of the methods that query the
database using features that \name supports. The features \name does
not currently support include the use of Rails \textit{scopes}, which
are essentially macros for queries, and the use of SQL strings for
methods other than \code{where}.

Finally, because \name performs type checking at runtime (see \S~\ref{subsec:databases}),
we must first load each benchmark before type checking it. We ran the type checker immediately
after loading a program and its associated type annotations.

%

\subsection{Results}
\label{subsec:results}


Table \ref{tab:type check-bench} summarizes our type checking
results. In the first group of columns, we list the number of type
checked methods and the total lines of code (computed with
\code{sloccount}) of these methods. The third column lists the number
of additional annotations we wrote for any global and instance variables
referenced in the method, as well as any methods called that were not
themselves selected for type checking.
The last
column in this group lists the number of type casts we
added. Many of these type casts were to the result of \code{JSON.parse},
which returns a nested \class{Hash}/\class{Array} data structure
depending on its string input. Most of the remaining casts are to
refine types after a conditional test; it may be possible to remove
these casts by adding support for occurrence
typing~\cite{Kent16}.
We further discuss type casts, in particular the reduced type casting
burden afforded by \comptypes, below.


\paragraph{Increased Type Checking Precision.}



Recall from \S~\ref{subsec:overview:hashes}
that \comptypes can potentially
reduce the need for programmer-inserted type casts. The next column
reports how many casts were needed using normal RDL (i.e., no
\comptypes). As shown, approximately 4.75$\times$ fewer casts were
needed when using \comptypes. This reflects the significantly 
increased precision afforded by \comptypes, which greatly reduces
the programmer's annotation burden.


\paragraph*{Performance.}

The next group of columns report performance. First we give the type
checking time as the median and semi-interquartile range (SIQR) of 11
runs on a 2017 MacBook Pro with a 2.3GHz i5 processor and 8GB RAM.  In
total, we type checked 132 methods in approximately 15 seconds, which
we believe to be reasonable. 
\textit{Discourse} took most of the total time (8 out of 15
seconds). The reason turned out to be a quirk of \textit{Discourse}'s
design: it creates a large number of methods on-the-fly when certain
constants are accessed. Type checking accessed those constants, hence
the method creation was included in the type checking time.



The next two columns show the performance overhead of the dynamic
checks inserted by \name. We selected a subset of each app's test
suite that directly tested the type checked methods, and ran
these tests without (``No Chk'') and with (``w/Chk'') the dynamic
checks. In aggregate (last row), checks add about 1.6\%
overhead, which is minimal.


\paragraph*{Errors Found.}

Finally, the last column lists the number of errors found in each
program. We were somewhat surprised to find any errors in large,
well-tested applications. We found three errors. In \textit{Code.org},
the \code{current_user} method was documented as returning a
\class{User}. We wrote a matching \name annotation, and \name
found that the returned expression---whose typing involved a
\comptype---has a hash type instead.  We notified the
\textit{Code.org} developers, and they acknowledged that this was an
error in the method documentation and made a fix.


 
In \textit{Journey}, \name found two errors. First, it
found a method
that referenced an undefined constant \code{Field}.  We notified the
developers, who fixed the bug by changing the constant to
\code{Question::Field}. This bug had arisen due to namespace changes.
Second, it found a method
that included a call with an argument \code{\{:action => prompt,
  ...\}} which is a hash mapping key \code{:action} to
\code{prompt}. The value \code{prompt} is supposed to be a string or
symbol, but as it has neither quotes nor begins with a colon, it is
actually a call to the \code{prompt} method, which returns an
array. The developers confirmed this bug.

When type checking the aforementioned methods in RDL (i.e., without comp
types), two out of three of the bugs are hidden by other type errors
which are actually false positives. These errors can be removed by
adding four type casts, which would then allow us to catch the true errors.
With CompRDL, however, we do not need any casts to find the errors.

\section{Related Work}
\label{sec:related-work}

\paragraph{Types For Dynamic Languages.}

There is a large body of research on adding static type systems to dynamic languages,
including Ruby~\cite{Furr:2009,ren:oops13,Ren}, Racket~\cite{Tobin-Hochstadt2006, Tobin-Hochstadt2008},
JavaScript~\cite{Anderson2005, Thiemann2005, Lerner2013}, and Python~\cite{Aycock2000,Ancona:2007}. 
To the best of our knowledge, this research does not use type-level computations.

Dependent typing systems for dynamic languages have been explored as well.
\citet{Ou:2004} formally model type-level computation along with effects for a dynamic language.
Other projects have sought to bring dependent types to existing dynamic languages, primarily in the form of
refinement types~\citep{Freeman:1991}, which are base types that are refined with expressive logical predicates.
Refinement types have been applied to Ruby~\cite{Kazerounian}, Racket~\cite{Kent16}, and
JavaScript~\cite{Chugh:2012,Vekris16}.
In contrast to \name, these systems focus on type checking methods which themselves have dependent types.
On the other hand, \name uses type-level computations only for non-type checked library methods,
allowing us to avoid checking \comptypes for equality or subtyping (\S~\ref{subsec:discussion}).
While sacrificing some expressiveness, this makes \name especially practical for real-world programs.

Turnstile~\cite{Chang:2017} is a metalanguage, hosted in Racket, for creating typed embedded languages.
It lets an embedded DSL author write their DSL’s type system using the host language macro system. 
There is some similarity to \name, where \comptypes manipulate standard RDL types. 
However, \name types are not executed as macros (which do not exist in Ruby), 
but rather in standard Ruby so they have full access to the environment, e.g., so 
the \textsf{joins} type signature can look up the DB schema.

\paragraph{Types For Database Queries.}
There have been a number of prior efforts to check the type safety of database queries. 
All of these target statically typed languages, an
important distinction from \name.

\citet{Chlipala:2010} presents Ur, a web-specific functional programming language. Ur uses type-level computations
over record types~\cite{Remy:1989} to type check programs that construct and run SQL queries.
Indeed, \name similarly uses type-level computations over finite hash types
(analogous to record types) to type check queries.
To the best of our knowledge, Ur focuses on computations over
records. In contrast,
\name supports arbitrary type-level computations targeting unchecked library methods, 
making \comptypes more easily extensible to checking new properties and
new libraries. As discussed in \S~\ref{sec:overview}, for example,
\comptypes can not only compute the schema of a joined table, but also check properties
like two joined tables having a declared Rails association. Further, \comptypes can
be usefully applied to many libraries beyond database queries (\S~\ref{sec:experiments}).

Similar to Ur, \citet{Gordon:ECOOP11} makes use of record types over embedded SQL tables. Using SMT-checked refinement types,
they can statically verify expressive data integrity constraints, such as the uniqueness of
primary keys in a table and the validation of data inserted into a table. In addition to the contrast
we draw with Ur regarding extensibility of types,
to the best of our knowledge,
this work does not include more intricate queries like  joins, which are supported in \name.

\paragraph*{New Languages for Database Queries.}

Domain-specific languages have long been used to write programs with
correct-by-construction, type safe queries.  \citet{Leijen1999} implement Haskell/DB, 
an embedded DSL that dynamically generates SQL in Haskell.  
\citet{Karakoidas:2015} introduce J\%, a Java extension for
embedding DSLs into Java in an extensible, type-safe, and
syntax-checked way.  \citet{Fowler:2013} use dependent types in
the language Idris to enforce safety protocols associated with common
web program features including database queries written in a DSL.

Language-integrated query is featured in languages like LINQ~\cite{Meijer:2006} and Links~\cite{Cooper:2006, Cheney:2014}.
This approach allows programmers to write database queries directly
within a statically-typed, general purpose language.

In contrast to new DSLs and language-integrated query, our focus in on bringing
type safety to an existing language and framework rather than
developing a new one.

\paragraph{Dependent Types.}

Traditional dependent type systems are exemplified by languages such as Coq~\cite{Pierce:SFold}, Agda~\cite{Norell2009}, and F*~\cite{Swamy:2016}. These languages provide powerful type
systems that allow programmers to prove expressive properties. However, such expressive types
 may be too heavyweight for a dynamic language like Ruby.
 As discussed in~\S~\ref{subsec:discussion}, 
our work has focused on applying a limited form of dependent types, 
where types depend on argument types and not arbitrary program values, 
resulting in a system that is practical for real-world Ruby programs.

Haskell allows for light dependent typing using the combination of 
singleton types~\cite{Eisenberg:2012} and type families~\cite{Chakravarty:2005}. 
\name's singleton types are similar to Haskell's, i.e., both lifting expressions to 
types, and \comptypes are analogous to anonymous type families. 
However, unlike Haskell, \name supports runtime evaluation during type checking, 
and thus does not require user-provided proofs. 

Scala supports \emph{path dependent types}, 
a limited form of type/term dependency in which types can depend on variables, 
but, as of Scala version 2, does not allow dependency on general terms~\cite{Amin:215280}. 
This allows for reasoning about database queries.
For example, the Scala library \textit{Slick}~\cite{slick}, much like our approach, 
allows users to write database queries in a domain specific language 
(a lifted embedding) and uses the query's AST to 
type check the query using Scala's path dependent types. 
Unlike \name, Scala's path dependent types 
do not allow the execution of the full host language during type computations.

\section{Conclusion}
\label{sec:conclusion}

We presented \name, a system for adding type signatures with
type-level computations, which we refer to as \comptypes, to Ruby
library methods. \name makes it possible to write \comptypes for
database queries, enabling us to type check such queries precisely.
\Comptype signatures can also be used for libraries over heterogeneous
hashes and arrays, and to treat strings as immutable when possible.
The increased precision of \comptypes can reduce the need for
manually inserted type casts, thereby reducing the programmer's burden
when type checking.
Since \comptype-annotated method bodies are not
themselves type checked, \name inserts run-time checks to ensure those
methods return their computed types.  We formalized \name as a core
 language \corelang and proved its type system 
sound.

We implemented \name on top of RDL, an existing type system for
Ruby. In addition to the features of \corelang, our implementation
includes run-time checks to ensure \comptypes that depend
on mutable state yield consistent types. Our implementation also
includes a termination checker for type-level code, and the type
signatures we developed perform weak updates to type certain mutable
methods.

Finally, we used \name to write \comptypes for several Ruby libraries
and two database query DSLs. Using these type signatures, we were able
to type check six popular Ruby apps and APIs, in the process
discovering three errors in our subject programs. We also found
that type checking with \comptypes required 4.75$\times$ fewer type
casts, due to the increased precision.  Thus, we believe that \name
represents a practical approach to precisely type checking
programs written in dynamic languages.

\begin{acks}                            
  Thanks to the anonymous reviewers for their helpful comments.
This research was supported in part by NSF CCF-1518844,
CCF-1846350, DGE-1322106, and Comunidad de Madrid 
as part of the program S2018/TCS-4339 (BLOQUES-CM) 
co-funded by EIE Funds of the European Union.
\end{acks}

\balance
\bibliography{arxiv-bibfile}

\clearpage
\appendix
\section{Appendix}
\label{sec:appendix}

This section contains the full definitions, static and dynamic semantics, and proof of soundness
for \corelang.

Figure~\ref{figure:lang-appendix} once again presents the syntax of \corelang, as well as some new auxiliary
definitions to be used for defining the semantics and proving soundness. Here we see for the first time
object instances \vinst{A}, which denote an instance of a class $A$.
The dynamic semantics rules are shown in Figure~\ref{figure:dyn-sem-appendix}. Most of the rules are standard. The  slightly more intricate rule is (\reContext), which takes a step 
within a subexpression, and it contains premises that differentiate it from (\reAppUD), the other context cases, and ensure that the context \context is the largest possible context for purposes of disambiguation.
Figure~\ref{figure:check-insertion-appendix} once again defines the dynamic check insertion rules.
Finally, Figure~\ref{figure:type-checking-appendix} defines the type checking rules, which are largely 
a simplification of the check insertion rules; it is for these type checking rules that we prove preservation
and progress. We prove soundness of the check insertion rules as a corollary of the soundness of the 
type checking rules.

\begin{figure*}[t!]
	\centering
	$$
	\begin{array}{rrcl}
	\texttt{values} \quad
	& \val
	& ::= & \vnil \spmid \vinst{A} \spmid \type \spmid \vtrue \spmid \vfalse
	\\
	\texttt{expressions} \quad
	& \expr
	& ::= & \val \spmid \var \spmid \eself \spmid \trec \spmid \eseq{\expr}{\expr} \\
	&     & \spmid & \enew{A} \spmid \eif{\expr}{\expr}{\expr} \spmid \eq{\expr}{\expr}\\
	&     & \spmid & \emethcall{\expr}{m}{\expr} \spmid \echeckmeth{\type}{\expr}{m}{\expr} 
	
	\\
	\texttt{method types} \quad
	& \methtype
	& ::= & \mthtype{\type}{\type} \\
	\texttt{TLC method types} \quad
	& \dmethtype
	& ::= & \methtype \spmid \depmethtype{\expr}{\type}{\expr}{\type}
	\\
	\textit{programs} \quad
	& \prog
	& ::= & \pdef{A.m}{\var}{\methtype}{\expr} \spmid \bidef{A.m}{\var}{\dmethtype} \spmid \prog ; \prog
	
	\end{array}
	$$
	
	$x\in\textrm{var IDs}$, $m\in\textrm{meth IDs}$, $A\in\textrm{class IDs}$
	
	$$
	\begin{array}{rrcl}
	\texttt{dyn env} \quad
	& \dynenv
	& : & \textnormal{var ids}\rightarrow \textnormal{values}
	\\
	\texttt{contexts} \quad
	& \context
	& ::= & \square  \spmid \emethcall{C}{m}{\expr} \spmid \emethcall{\val}{m}{\context} \spmid \echeckmeth{\type}{\context}{m}{\expr} \spmid \echeckmeth{\type}{\val}{m}{\context} \spmid \eseq{\context}{\expr} \\
	&     & \spmid & \eif{\context}{\expr}{\expr}  \spmid \eq{\context}{\expr} \spmid \eq{\val}{\context}
	\\
	\texttt{stack} \quad
	& \stack
	& ::= & \cdot \spmid \stacktop{(\dynenv,\context)}{S}
	\\
	\texttt{type stack} \quad
	& \typestack
	& ::= & \cdot \spmid \stacktop{\tstackelement{\Gamma}{\type}{\type}}{\typestack}
	\\
	\texttt{typ env} \quad
	& \Gamma, \Delta
	& : & \textnormal{var ids}\rightarrow \textnormal{base types}
	\\
	\texttt{class table} \quad
	& \classtable
	& : & \textnormal{class ids}\rightarrow \textnormal{meth ids}\rightarrow \textnormal{types}
	\\
	\texttt{objects} \quad
	& \objrecord
	& : & \textnormal{objects}
	\\
	\texttt{object instance} \quad
	& \vinst{\cdot}
	& : & A\rightarrow\objrecord
	\\
	\\
	\texttt{method sets} \quad
	& \userdef
	& : & \textnormal{set of user-defined methods}
	\\
	& \builtin
	& : & \textnormal{set of library methods}
	\\
	& & & \textnormal{where } \builtin\cap\userdef = \emptyset
	\end{array}
	$$
	\tnil, \tobj, \tbool, \ttrue, \tfalse, and \tuniv are all presumed to be class IDs $A$.
	Subtyping is defined as \subtype{\tnil}{A}, \subtype{A}{A}, and \subtype{A}{A\sqcup A'} for all $A, A'$. $A\sqcup A'$ is the least upper bound of types $A$ and $A'$. 
	
	\caption{\corelang and auxiliary definitions.}
	\label{figure:lang-appendix}
	
\end{figure*}

\begin{figure*}
	\centering
	\judgementHead{Dynamic semantics}{\eeval{\dyn{\dynenv}{\expr}{\stack}}{\dyn{\dynenv'}{\expr'}{\stack'}}}
	
	\begin{tabular}{lc}
		(\reVar)     & \eeval{\dyn{\dynenv}{\var}{\stack}}{\dyn{\dynenv}{\dynenv(\var)}{\stack}} \\
		(\reSelf)    & \eeval{\dyn{\dynenv}{\eself}{\stack}}{\dyn{\dynenv}{\dynenv(\eself)}{\stack}} \\
		(\reTSelf)    & \eeval{\dyn{\dynenv}{\trec}{\stack}}{\dyn{\dynenv}{\dynenv(\trec)}{\stack}} \\
		(\reSeq)     & \eeval{\dyn{\dynenv}{\val;\expr}{\stack}}{\dyn{\dynenv}{\expr}{\stack}} \\
		(\reNew)     & \eeval{\dyn{\dynenv}{\enew{A}}{\stack}}{\dyn{\dynenv}{\vinst{A}}{\stack}} \\
		(\reEIfTrue)  & \eeval{\dyn{\dynenv}{\eif{\val}{\expr_2}{\expr_3}}{\stack}}{\dyn{\dynenv}{\expr_2}{\stack}} \\ 
		& if \val is not \vnil or \vfalse \\
		(\reEIfFalse) & \eeval{\dyn{\dynenv}{\eif{\val}{\expr_2}{\expr_3}}{\stack}}{\dyn{\dynenv}{\expr_3}{\stack}} \\
		& if $\val = \vnil$ or $\val = \vfalse$ \\
		(\reEEqTrue) & \eeval{\dyn{\dynenv}{\eq{\val_1}{\val_2}}{\stack}}{\dyn{\dynenv}{\vtrue}{\stack}} \\ 
		  & if $\val_1$ and $\val_2$ are equivalent \\
		(\reEEqFalse) & \eeval{\dyn{\dynenv}{\eq{\val_1}{\val_2}}{\stack}}{\dyn{\dynenv}{\vfalse}{\stack}}\\
		& if $\val_1$ and $\val_2$ are not equivalent\\
		(\reAppUD) & \begin{tabular}[t]{@{}c@{}}
			\eeval{\dyn{\dynenv}{\context[\val_r.m(\val)]}{\stack}}{\dyn{[\eself \mapsto \val_r, \var \mapsto \val]}{\expr}{\stacktop{\dstackelement{\dynenv}{\context}}{\stack}}} \\
			if $\typeof(\val_r) = A$
			 and $A.m\in\userdef$ and \defof{A.m}{\var}{e}
		\end{tabular} \\
		(\reAppBI) & \begin{tabular}[t]{@{}c@{}}
			\eeval{\dyn{\dynenv}{\echeckmeth{\type}{\val_r}{m}{\val}}{\stack}}{\dyn{\dynenv}{\val'}{\stack}} \\
			if $\typeof(\val_r) = A$
			and $A.m\in\builtin$ and $\val' = \call{A.m}{\val_r}{\val}$ and \subtype{\typeof(\val')}{\type}
		\end{tabular} \\
		
		(\reRet)     & \eeval{\dyn{\dynenv'}{\val}{\stacktop{\dstackelement{\dynenv}{\context}}{\stack}}}{\dyn{\dynenv}{\context[\val]}{\stack}}
	\end{tabular}
	\begin{flushleft}(\reContext)\end{flushleft}
	$$
	\inference{\eeval{\dyn{\dynenv}{\expr}{\stack}}{\dyn{\dynenv'}{\expr'}{\stack'}} 
		&& \nexists \val, \val_r, A, m. (\expr = \val_r.m(\val)\land \typeof(\val_r) = A \land A.m\in\userdef) 
		\\ \nexists \val. \expr = \val 
		&& \nexists \context', e''. \expr = \context'[\expr'']}
	{\eeval{\dyn{\dynenv}{\context[\expr]}{\stack}}{\dyn{\dynenv'}{\context[\expr']}{\stack'}}}
	$$
	
	\begin{flushleft}
		
		where \defof{A.m}{\var}{\expr} if there exists a definition \pdef{A.m}{\var}{\methtype}{\expr}. Additionally, \call{A.m}{\val}{\val} is our way of dispatching a library method, where $A.m$ identifies the method, the first \val is the receiver of the method call, and the second \val is the argument of the method call. It returns a value $\val_r$, the value returned by the method call. Finally, we use \typeof(\val), where \typeof(\vnil) = \tnil, \typeof(\type) = \tuniv, \typeof(\vtrue) = \ttrue, \typeof(\vfalse) = \tfalse, and \typeof(\vinst{A}) = $A$.
		
	\end{flushleft}
	\caption{Dynamic semantics of \corelang}
	\label{figure:dyn-sem-appendix}
\end{figure*}

\begin{figure*}
	\centering
	\judgementHead{Type Checking and Check Insertion Rules.}{\tcheck{\classtable}{\Gamma}{\expr}{\expr}{\type}}
	
	$$
	\inference{}{\tcheck{\classtable}{\Gamma}{\vnil}{\vnil}{\tnil}}[\cNil]
	\quad
	\inference{}{\tcheck{\classtable}{\Gamma}{\enew{A}}{\enew{A}}{A}}[\cNew]
	\quad
	\inference{}{\tcheck{\classtable}{\Gamma}{\type}{\type}{\tuniv}}[\cType]
	$$
	
	$$
	\inference{}{\tcheck{\classtable}{\Gamma}{\vtrue}{\vtrue}{\ttrue}}[\cTrue]
	\quad
	\inference{}{\tcheck{\classtable}{\Gamma}{\vfalse}{\vfalse}{\tfalse}}[\cFalse]
	\quad
		\inference{}{\tcheck{\classtable}{\Gamma}{\vinst{A}}{\vinst{A}}{A}}[\cObj]
	$$
	
	$$
	\inference{\Gamma(\var) = \type}{\tcheck{\classtable}{\Gamma}{\var}{\var}{\type}}[\cVar]
	\quad
	\inference{\tcheck{\classtable}{\Gamma}{\expr_1}{\expr_1'}{\type_1} \\
		\tcheck{\classtable}{\Gamma_1}{\expr_2}{\expr_2'}{\type_2}}
	{\tcheck{\classtable}{\Gamma}{\eq{\expr_1}{\expr_2}}{\eq{\expr_1'}{\expr_2'}}{\tbool}}[\cEq]
	\quad
		\inference{\tcheck{\classtable}{\Gamma}{\expr_1}{\expr_1'}{\type_1} \\
		\tcheck{\classtable}{\Gamma_1}{\expr_2}{\expr_2'}{\type_2}}
	{\tcheck{\classtable}{\Gamma}{\eseq{\expr_1}{\expr_2}}{\eseq{\expr_1'}{\expr_2'}}{\type_2}}[\cSeq]
	$$	
	
	$$
	\inference{
		\tcheck{\classtable}{\Gamma}{\expr_1}{\expr_1'}{\type_1} &&
		\tcheck{\classtable}{\Gamma}{\expr_2}{\expr_2'}{\type_2} &&
		\tcheck{\classtable}{\Gamma}{\expr_3}{\expr_3'}{\type_3}}
	{\tcheck{\classtable}{\Gamma}{\eif{\expr_1}{\expr_2}{\expr_3}}{\eif{\expr_1'}{\expr_2'}{\expr_3'}}{(\type_2\sqcup\type_3)}}[\cIf]
	\quad
	\inference{
		\cteval{\classtable}{A.m}{\mthtype{\type_1}{\type_2}} &&
		A.m\in\userdef \\
		\subtype{\type_x}{\type_1} \\
		\tcheck{\classtable}{\Gamma}{\expr}{\expr'}{A} &&
		\tcheck{\classtable}{\Gamma}{\expr_x}{\expr_x'}{\type_x}
	}
	{\tcheck{\classtable}{\Gamma}{\emethcall{\expr}{m}{\expr_x}}{\emethcall{\expr'}{m}{\expr_x'}}{\type_2}}[\cappBaseUD]
	$$
	
	$$
	\inference{
		\cteval{\classtable}{A.m}{\mthtype{\type_1}{\type_2}} \\
		A.m\in\builtin &&
		\subtype{\type_x}{\type_1} \\
		\tcheck{\classtable}{\Gamma}{\expr}{\expr'}{A} \\
		\tcheck{\classtable}{\Gamma}{\expr_x}{\expr_x'}{\type_x}
	}
	{\tcheck{\classtable}{\Gamma}{\emethcall{\expr}{m}{\expr_x}}{\echeckmeth{\type_2}{\expr'}{m}{\expr_x'}}{\type_2}}[\cappBaseBI]
	\quad
	\inference{
		\cteval{\classtable}{A.m}\depmethtype{\expr_{t1}}{\type_{t1}}{\expr_{t2}}{\type_{t2}} &&
		A.m\in \builtin \\
		\tcheck{\classtable}{\Gamma}{\expr}{\expr'}{A} &&
		\tcheck{\classtable}{\Gamma}{\expr_x}{\expr_x'}{\type_x} \\
		\tcheck{\unrefine{\classtable}}{\envbind{\var}{\tuniv},\envbind{\trec}{\tuniv}}{\expr_{t1}}{\expr_{t1}'}{\tuniv} \\
		\eevalmulti{\dyn{\subst{\var}{\type}\subst{\trec}{A}}{\expr_{t1}'}{\cdot}}{\dyn{\dynenv_1}{\type_1}{\cdot}} && 
		\subtype{\type_x}{\type_1} \\
		\tcheck{\unrefine{\classtable}}{\envbind{\var}{\tuniv},\envbind{\trec}{\tuniv}}{\expr_{t2}}{\expr_{t2}'}{\tuniv} \\
		\eevalmulti{\dyn{\subst{\var}{\type}\subst{\trec}{A}}{\expr_{t2}'}{\cdot}}{\dyn{\dynenv_2}{\type_2}{\cdot}} \\
	}
	{\tcheck{\classtable}{\Gamma}{\emethcall{\expr}{m}{\expr_x}}{\echeckmeth{\type_2}{\expr'}{m}{\expr_x'}}{\type_2}}[\cappComp]
	$$ 
	
	\caption{Type checking and check insertion rules for \corelang.}
	\label{figure:check-insertion-appendix}
\end{figure*}

\begin{figure*}
	\centering
	\judgementHead{Type checking rules}{\teval{\classtable}{\Gamma}{\expr}{\type}}
	
	$$
	\inference{}{\teval{\classtable}{\Gamma}{\vnil}{\tnil}}[\tNil]
	\quad
	\inference{}{\teval{\classtable}{\Gamma}{\vinst{A}}{A}}[\tObj]
	\quad
	\inference{\Gamma(\eself)=\type}{\teval{\classtable}{\Gamma}{\eself}{\type}}[\tSelf]
	\quad
	\inference{}{\teval{\classtable}{\Gamma}{\vtrue}{\ttrue}}[\tTrue]	
	$$
	\trspace
	$$
	\inference{}{\teval{\classtable}{\Gamma}{\vfalse}{\tfalse}}[\tFalse]
	\quad
	\inference{}{\teval{\classtable}{\Gamma}{\type}{\tuniv}}[\tType]
	\quad
	\inference{\Gamma(\var) = \type}{\teval{\classtable}{\Gamma}{\var}{\type}}[\tVar]
	$$
	\trspace
	$$
	\inference{\teval{\classtable}{\Gamma}{\expr_1}{\type_1} \\
		\teval{\classtable}{\Gamma_1}{\expr_2}{\type_2}}
	{\teval{\classtable}{\Gamma}{\eq{\expr_1}{\expr_2}}{\tbool}}[\tEq]
	\quad
	\inference{\Gamma(\trec)=\type}{\teval{\classtable}{\Gamma}{\trec}{\type}}[\tTSelf]
	\quad
	\inference{\teval{\classtable}{\Gamma}{\expr_1}{\type_1} \\
		\teval{\classtable}{\Gamma_1}{\expr_2}{\type_2}}
	{\teval{\classtable}{\Gamma}{\eseq{\expr_1}{\expr_2}}{\type_2}}[\tSeq]
	\quad
	\inference{}{\teval{\classtable}{\Gamma}{\enew{A}}{A}}[\tNew]
	$$
	\trspace
	$$
	\inference{\teval{\classtable}{\Gamma}{\expr_1}{\type_1} &&
	\teval{\classtable}{\Gamma}{\expr_2}{\type_2} \\
	\teval{\classtable}{\Gamma}{\expr_3}{\type_3}}
{\teval{\classtable}{\Gamma}{\eif{\expr_1}{\expr_2}{\expr_3}}{(\type_2\sqcup\type_3)}}[\tIf]
	\inference{\teval{\classtable}{\Gamma}{\expr_0}{A} &&
		A.m\in\userdef \\
		\teval{\classtable}{\Gamma}{\expr_1}{\type} \\
		\cteval{\classtable}{A.m}{\mthtype{\type_1}{\type_2}} &&
		\subtype{\type}{\type_1} \\
	}
	{\teval{\classtable}{\Gamma}{\emethcall{\expr_0}{m}{\expr_1}}{\type_2}}[\tappBase]
	\quad
	\inference{\teval{\classtable}{\Gamma}{\expr_0}{A} \\
    A.m\in\builtin \\
}
{\teval{\classtable}{\Gamma}{\echeckmeth{\type}{\expr_0}{m}{\expr_1}}{\type}}[\tappComp]
	$$
	
	\caption{Type checking rules for \corelang.}
	\label{figure:type-checking-appendix}
\end{figure*}

\subsection{Soundness}



We prove soundness of the type system of \corelang by first proving preservation and then progress.
The proof of preservation is the more involved step here, and it requires a number of preliminary definitions.
First, we define a notion of consistency between the type and dynamic environments:

\begin{definition}[Environmental consistency]
\textit{Type environment $\Gamma$ is consistent with dynamic environment \dynenv, written \consistent{\Gamma}{\dynenv}, if for all variables \var, $\var\in dom(\dynenv)$ if and only if $\var\in dom(\Gamma)$, and for all $x\in dom(\dynenv)$ there exists \type\ such that} \teval{\classtable}{\Gamma}{\dynenv(x)}{\type} and \subtype{\type}{\Gamma(x)}.
\end{definition}

We will use the notation \typeof(\val), where \typeof(\vnil) = \tnil, 
\typeof(\vtrue) = \ttrue,
\typeof(\vfalse) = \tfalse,
$\typeof(\vinst{A}) = A$, and \typeof(\type) = \tuniv, for any $A$. 

\textit{On blame:} Our type system does not prevent invoking a method on a value that is \vnil. Additionally, runtime evaluation can fail if an inserted dynamic
check fails.
In order to retain soundness of our system, we add dynamic semantics rules which step to \textit{blame} in these cases, where \teval{\classtable}{\Gamma}{blame}{\tnil} for any $\Gamma$. Additionally, we add rules which take a step to $blame$ whenever a subexpression takes a step to $blame$. We omit the rules here for brevity.

Because we make use of a stack in our dynamic semantics, a standard type preservation theorem
which says that we always step to an expression which has the same type (or subtype) will not suffice.
Rules (\reAppUD) and (\reRet) push and pop from the stack. In these cases, an expression $\expr$ may have an entirely different type than the expression that it steps to, $\expr'$. To account for this we incorporate a notion of a \textit{type stack} \typestack\ to mirror the runtime stack, which is defined Figure~\ref{figure:lang-appendix}. As an example, suppose we want to apply preservation to $\context[\val_1.m(\val_2)]$. The type checking judgment is \teval{\classtable}{\Gamma}{\context[\val_1.m(\val_2)]}{\type'}. Because the dynamic semantics rule (\reAppUD) pushes the current environment and context onto the stack, we will push the current typing judgment on to the type stack. Specifically, we will push an element of the form \tstackelement{\Gamma}{\type}{\type'}, where $\Gamma$ is the environment of the current typing judgment; $\type'$ is the type of the surrounding context; and \type is the type of expression $\val_1.m(\val_2)$, i.e., the type that the method must return. 

With this type stack, we can now define what it means for a type to be
a subtype of the type stack, which is the crucial preservation invariant we will prove:

\begin{definition}[Stack subtyping]
\subtype{\type_0}{\stacktop{\tstackelement{\Gamma}{\type}{\type'}}{\typestack}} if \subtype{\type_0}{\type}.
\end{definition}

\begin{definition}[Stack consistency]
\textit{Type stack element }\tstackelement{\Gamma}{\type}{\type'} is consistent with dynamic stack element \dstackelement{\dynenv}{\context}, written \consistent{\tstackelement{\Gamma}{\type}{\type'}}{\dstackelement{\dynenv}{\context}}, if \consistent{\Gamma}{\dynenv} and \teval{\classtable}{\Gamma\subst{\square}{\type}}{\context}{\type'} (Here we abuse notation and treat $\square$ as if its a variable.)

\textit{Type stack} \typestack  is consistent with dynamic stack \stack, written \consistent{\typestack}{\stack}, is defined inductively as
\begin{itemize}
\item[1.] \consistent{\cdot}{\cdot}
\item[2.] \consistent{\stacktop{\tstackelement{\Gamma}{\type}{\type'}}{\typestack}}{\stacktop{\dstackelement{\dynenv}{\context}}{\stack}} if
\begin{itemize}
\item[(a)] \consistent{\tstackelement{\Gamma}{\type}{\type'}}{\dstackelement{\dynenv}{\context}}
\item[(b)] \consistent{\typestack}{\stack}
\item[(c)] \subtype{\type'}{\typestack} if $\typestack\neq\cdot$
\end{itemize}
\end{itemize}
\end{definition}

We will also make use of a notion of class table validity, which tells us that a class table maps fields and methods
to the "correct" types:

\begin{definition}[Class table validity]\label{def:class-table-validity}
Let $A.m$ be an arbitrary method. We say \valid(\classtable) if
\begin{itemize}
\item[1.] if $A.m\in\userdef$, then $A.m\in dom(\classtable)$, \cteval{\classtable}{A.m}{\mthtype{\type_1}{\type_2}} for some $\type_1, \type_2$, and there exists a single method definition \pdef{A.m}{\var}{\mthtype{\type_1}{\type_2}}{\expr} such that \teval{\classtable}{[\eself \mapsto A, \var \mapsto \type_1]}{\expr}{\type_2'} for some $\type_2'$ where \subtype{\type_2'}{\type_2}. 

\item[2.] if $A.m\in\builtin$, then $A.m\in dom(\classtable)$, \cteval{\classtable}{A.m}{\methtype} for some $\type_m$, and there exists a single method declaration \bidef{A.m}{\var}{\methtype} such that \methtype = \mthtype{\type_1}{\type_2} or \methtype = \depmethtype{\expr_{t1}}{\type_{1}}{\expr_{t1}}{\type_{2}} for some $\type_1, \type_2, \expr_{t1}, \expr_{t2}$. 

\item[3.] For all $A'$ such that \subtype{A'}{A}, if $\classtable(A'.m) = \mthtype{\type_1'}{\type_2'}$ and 
$\classtable(A.m) = \mthtype{\type_1}{\type_2}$ then \subtype{\type_1}{\type_1'} and \subtype{\type_2'}{\type_2}.

\end{itemize}
\end{definition}

See section~\ref{subsec:prog-type-check} for the programming type checking rules, and a discussion of
how to construct and check the validity of a class table. 

Finally, we make use of the following lemmas in our proofs of soundness:

\begin{lemma}[Contextual substitution]\label{lemma:contextual-substitution}
\textit{If}
\[
\inference{\teval{\classtable}{\Gamma}{\expr}{\type'}\\ \vdots}
{\teval{\classtable}{\Gamma_C}{\context[\expr]}{\type_C}},
\]
\textit{then} \teval{\classtable}{\Gamma_C\subst{\square}{\type'}}{C}{\type_C}.
\end{lemma}

\begin{lemma}[Substitution]\label{lemma:substitution}
If

\begin{itemize}
\item[1.] \teval{\classtable}{\Delta\subst{\square}{\type_C}}{C}{\type_C'}
\item[2.] \teval{\classtable}{\Delta}{\expr}{\type}
\item[3.] \subtype{\type}{\type_C}
\end{itemize}
then \teval{\classtable}{\Delta}{\context[\expr]}{\type_C''} where \subtype{\type_C''}{\type_C'}.
\end{lemma}


With the above definitions and lemmas, we can finally state our preservation theorem:

\begin{theorem}[Preservation]
\textit{If}
\begin{itemize}
\item[(1)] \eeval{\dyn{\dynenv}{\expr}{\stack}}{\dyn{\dynenv'}{\expr'}{\stack'}}
\item[(2)] \teval{\classtable}{\Gamma}{\expr}{\type}
\item[(3)] \subtype{\type}{\typestack}
\item[(4)] \consistent{\Gamma}{\dynenv}
\item[(5)] \consistent{\typestack}{\stack}
\item[(6)] \valid(\classtable)
\end{itemize}
\textit{then there exists $\Delta, \typestack', \type'$ such that}
\begin{itemize}
\item[(a)] \teval{\classtable}{\Delta}{\expr'}{\type'}		
\item[(b)] \subtype{\type'}{\typestack'}
\item[(c)] \consistent{\Delta}{\dynenv'}
\item[(d)] \consistent{\typestack'}{\stack'}
\item[(e)] If $\stack = \stack'$ then \subtype{\type'}{\type} and $\Delta = \Gamma,\Gamma'$ for some $\Gamma'$
\end{itemize}
\end{theorem}

(1) and (2) are standard: they say that some expression \expr takes a step, and that \expr is well typed.
Conclusion (a) states that $\expr'$ is also well typed.
(3) says that the type of \expr is a subtype of the type stack,
and (b) says the same of $\expr'$; this is the crux of the type preservation proof, and as explained above,
it must be phrased in this way in order to account for the use of a stack, which allows us to step to
expressions with completely different types. Notice that (e) also tells us that if the stack goes unchanged,
then \subtype{\type'}{\type}; this is much closer to the standard
statement of preservation. (e) also gives us that if the stack goes unchanged, then the new type
environment is an extension of the old one.

(4) gives us consistency between type and dynamic environments, 
and (5) gives us consistency between the type stack and stack,
the corresponding conclusions (c) and (d) respectively give us the
same for the new environments.

(6) gives us validity of the class table (a corresponding conclusion is
unnecessary since the class table goes unchanged). 

Finally, we proceed with the proof of preservation.

\begin{proof}
By induction on \eeval{\dyn{\dynenv}{\expr}{\stack}}{\dyn{\dynenv'}{\expr'}{\stack'}}
\begin{itemize}

\item Case (\reSelf). By assumption we have
\begin{itemize}
\item[(1)] \eeval{\dyn{\dynenv}{\eself}{\stack}}{\dyn{\dynenv}{\dynenv(\eself)}{\stack}} by (\reSelf)
\item[(2)] \teval{\classtable}{\Gamma}{\eself}{\type}
\item[(3)] \subtype{\Gamma(\eself)}{\typestack}
\item[(4)] \consistent{\Gamma}{\dynenv}
\item[(5)] \consistent{\typestack}{\stack}
\item[(6)] \valid(\classtable)
\end{itemize}
Since (2) must have been derived by (\tSelf), by inversion of this rule we have that $\type = \Gamma(\eself)$.
Let $\Delta = \Gamma$ and $\typestack' = \typestack$. By equality of $\Delta$ and $\Gamma$, (2), (4), and the definition of environmental consistency, there exists $\type'$ such that \teval{\classtable}{\Delta}{\dynenv(\eself)}{\type'} and \subtype{\type'}{\Delta(\eself)}. Then (a) holds since $\dynenv(\eself)$ is well typed. (b) holds since $\subtype{\type'}{\Delta(\eself)} = \subtype{\Gamma(\eself)}{\typestack}$ by (3). Because \subtype{\type'}{\Gamma(\eself)} and $\Delta = \Gamma,\emptyset$, we have (e). Finally, (c) holds by (4), and (d) holds by (5).
\item Case (\reVar), (\reTSelf). Similar to (\reSelf) case.

\item Case (\reNew). By assumption we have
\begin{itemize}
\item[(1)] \eeval{\dyn{\dynenv}{\enew{A}}{\stack}}{\dyn{\dynenv}{\vinst{A}}{\stack}}  by (\reNew)
\item[(2)] \teval{\classtable}{\Gamma}{\enew{A}}{A} by (\tNew)
\item[(3)] \subtype{A}{\typestack}
\item[(4)] \consistent{\Gamma}{\dynenv}
\item[(5)] \consistent{\typestack}{\stack}
\item[(6)] \valid(\classtable)
\end{itemize}

Let $\Delta = \Gamma$, $\typestack' = \typestack$, and $\type' = \type$. Then we get (a) from
(\tObj), (b) from (3), (c) from (4), (d) from (5), and (e) immediately by definition of $\type'$ and $\Delta$. 

\item Case (\reSeq). Trivial.
\item Case (\reEEqTrue), (\reEEqFalse). Trivial. 
\item Case (\reEIfTrue). By assumption we have
\begin{itemize}
\item[(1)] \eeval{\dyn{\dynenv}{\eif{\val}{\expr_1}{\expr_2}}{\stack}}{\dyn{\dynenv}{\expr_1}{\stack}} 
\item[(2)] \teval{\classtable}{\Gamma}{\eif{\val}{\expr_1}{\expr_2}}{\type}
\item[(3)] \subtype{\type}{\typestack}
\item[(4)] \consistent{\Gamma}{\dynenv}
\item[(5)] \consistent{\typestack}{\stack}
\item[(6)] \valid(\classtable)
\end{itemize}

By (2) and inversion of (\tIf), there exits $\type_v, \type_1, \type_2$ such that:
\begin{itemize}
\item[(7)] \teval{\classtable}{\Gamma}{\val}{\type_v} 
\item[(8)] \teval{\classtable}{\Gamma}{\expr_1}{\type_1}
\item[(9)] \teval{\classtable}{\Gamma}{\expr_2}{\type_2}
\item[(10)] $\type = \type_1\cup\type_2$
\end{itemize}

Let $\Delta = \Gamma$, $\typestack' = \typestack$, and $\type' = \type_1$. Then from (8) we trivially have (a). Additionally, $\type' = \subtype{\type_1}{\type_1\sqcup\type_2}$, and so \subtype{\type'}{\typestack'} giving us (b). Finally, we have (c) by (4), (d) by (5), and (e) by the fact that \subtype{\type'}{\type} and $\Delta = \Gamma, \emptyset$.

\item Case (\reEIfFalse). Similar to (\reEIfTrue) case.

\item Case (\reContext). By assumption we have:
\begin{itemize}
\item[(1)] \eeval{\dyn{\dynenv}{\context[\expr]}{\stack}}{\dyn{\dynenv'}{\context[\expr']}{\stack'}} where
\begin{itemize}
\item[(1a)] \eeval{\dyn{\dynenv}{\expr}{\stack}}{\dyn{\dynenv'}{\expr'}{\stack'}}
\item[(1b)] $\neg(\expr=\val_r.m(v)\land \typeof(\var_r) = A \land A.m\in\userdef)$ for some \val, $\val_r$, $A$, $m$
\item[(1c)] $\expr\neq\val$ for some \val
\item[(1d)] $\expr\neq \context'[\expr']$ for some $\context'$, $\expr'$
by (\reContext)
\end{itemize}
\item[(2)] \teval{\classtable}{\Gamma}{\context[\expr]}{\type}
\item[(3)] \subtype{\type}{\typestack}
\item[(4)] \consistent{\Gamma}{\dynenv}
\item[(5)] \consistent{\typestack}{\stack}
\item[(6)] \valid(\classtable)
\end{itemize}
First, note that we must have $\stack' = \stack$, because the only cases where this would not happen would be if (\reAppUD) or (\reRet) were used to derive (1a), and this cannot be the case given (1b) and because (\reRet) only applies to top-level values and thus can't apply to a context. 
Now note that since (2) gives us $\teval{\classtable}{\Gamma}{\context[\expr]}{\type}$,
by inversion, there should exist $\type_e$ so that 

\begin{itemize}
\item[(7)] \teval{\classtable}{\Gamma}{\expr}{\type_e}
\end{itemize}

Let $\typestack_e$ be a type stack such that \consistent{\typestack_e}{\stack}, \subtype{\type_e}{\typestack_e}, and all type environments in $\typestack_e$ are the same as those in \typestack; it is straightforward to construct such a $\typestack_e$ from the existing \typestack. Then, by (1a), (7), \subtype{\type_e}{\typestack_e}, (4), \consistent{\typestack_e}{\stack}, and (6), we satisfy the premises of the preservation theorem. Therefore, applying the inductive hypothesis, there exists $\Delta_e, \typestack_e', \type_e'$ such that:

\begin{itemize}
\item[($\textnormal{a}_i$)] \teval{\classtable}{\Delta_e}{\expr'}{\type_e'}
\item[($\textnormal{b}_i$)] \subtype{\type_e'}{\typestack_e'}
\item[($\textnormal{c}_i$)] \consistent{\Delta_e}{\dynenv'}
\item[($\textnormal{d}_i$)] \consistent{\typestack_e'}{\stack'}
\item[($\textnormal{e}_i$)] If $\stack = \stack'$ then \subtype{\type_e'}{\type_e} and $\Delta_e = \Gamma,\Gamma'$ for some $\Gamma'$
\end{itemize}

Let $\Delta = \Delta_e$ and $\typestack' = \typestack$. Now, because $\stack' = \stack$, by (g$_i$) we have $\Delta = \Gamma,\Gamma'$. 
By Lemma~\ref{lemma:contextual-substitution}
we have \teval{\classtable}{\Gamma\subst{\square}{\type_e}}{\context}{\type},
and by the weakening lemma, we also have \teval{\classtable}{\Delta\subst{\square}{\type_e}}{\context}{\type}. By (e$_i$) we also have \subtype{\type_e'}{\type_e}. These, along with (a$_i$) and the substitution lemma~\ref{lemma:substitution}, give us that \teval{\classtable}{\Delta}{C[\expr']}{\type'} for some $\type'$ where \subtype{\type'}{\type}. This immediately gives us (a), and because \subtype{\type'}{\type} and \subtype{\type}{\typestack} by (3) we get (b). Because $\Delta = \Gamma,\Gamma'$ and \subtype{\type'}{\type}, we get (e). We get (c) from (c$_i$). (d) comes from (5) and the fact that $\stack' = \stack$. 

\item Case (\reAppBI). By assumption we have
\begin{itemize}	
	\item[(1)] \eeval{\dyn{\dynenv}{\echeckmeth{\type_{res}}{\val_r}{m}{\val_1}}{\stack}}{\dyn{\dynenv}{\val}{\stack}} where
	\begin{itemize}
		\item[(1a)] $\typeof(\val_r) = A_{rec}$		
		\item[(1b)] $A_{rec}.m\in\builtin$
		\item[(1c)] $\val = \call{A_{rec}.m}{\val_r}{\val_1}$
		\item[(1d)] \subtype{\typeof(\val)}{\type_{res}}
	\end{itemize}
	by (\reAppBI).
	\item[(2)] \teval{\classtable}{\Gamma}{\echeckmeth{\type_{res}}{\val_r}{m}{\val_1}}{\type}
	\item[(3)] \subtype{\type}{\typestack}
	\item[(4)] \consistent{\Gamma}{\dynenv}
	\item[(5)] \consistent{\typestack}{\stack}
	\item[(6)] \valid(\classtable)
\end{itemize}

Because of (1b), we know that (2) must have been derived by (\tappComp). Instantiating this rule with the
obvious bindings, we get that $\type = \type_{res}$. Let $\Delta = \Gamma$, $\typestack' = \typestack$,
and $\type' = \typeof(\val)$. 
If \val is \vnil, we get (a) immediately by (1d), definition of \typeof, and (\tNil); a similar argument holds
for all other potential values of \val. In all cases, we get (b) from (1d) and (3). We get (c) from (4), (d) from (5), and (e) from (1d) and definition of $\Delta$. 

\item Case (\reAppUD). By assumption we have

\begin{itemize}
\item[(1)]\eeval{\dyn{\dynenv}{\context[\emethcall{\val_r}{m}{\val}]}{\stack}}{\dyn{[\eself \mapsto \val_r, \var \mapsto \val]}{\expr}{\stacktop{\dstackelement{\dynenv}{\context}}{\stack}}} where
\begin{itemize}
\item[(1a)] $\typeof(\val_r) = A_{rec}$
\item[(1b)] $A_{rec}.m\in\userdef$
\item[(1c)] \defof{A_{rec}.m}{\var}{\expr}
\end{itemize}
\item[(2)] \teval{\classtable}{\Gamma}{\context[\emethcall{\val_r}{m}{\val}]}{\type_C}
\item[(3)] \subtype{\type_C}{\typestack}
\item[(4)] \consistent{\Gamma}{\dynenv} 
\item[(5)] \consistent{\typestack}{\stack}
\item[(6)] \valid(\classtable)
\end{itemize}

Noting the type checking rules for each context case, we know (2) must have been derived by some rule of the form:

$$\inference{\teval{\classtable}{\Gamma}{\emethcall{\val_r}{m}{\val}}{\type_m} \\ \vdots}{\teval{\classtable}{\Gamma}{\context[\emethcall{\val_r}{m}{\val}]}{\type_C}}$$

From this and the contextual substitution lemma~\ref{lemma:contextual-substitution}, 
we know \teval{\classtable}{\Gamma\subst{\square}{\type_m}}{\context}{\type_C}. Additionally, by inversion, we have 
\begin{itemize}
\item[(7)] \teval{\classtable}{\Gamma}{\emethcall{\val_r}{m}{\val}}{\type_m} 
\end{itemize}

By (1a), definition of \typeof, and the value type checking rules, we know that
\teval{\classtable}{\Gamma}{\val_r}{A_{rec}}.
Because by (1b) we know that $A_{rec}.m\in\userdef$ and  
$\userdef$ and $\builtin$ are by definition disjoint, 
the type checking judgment (7) must have been derived from rule \tappBase.

By instantiation of  rule \tappBase with $\type_m = \type_2$ we get  

$$
	\inference{\teval{\classtable}{\Gamma}{\val_r}{A_{rec}} &&
	A.m\in\userdef \\
	\teval{\classtable}{\Gamma}{\val}{\type_{arg}} \\
	\cteval{\classtable}{A.m}{\mthtype{\type_1}{\type_2}} &&
	\subtype{\type_{arg}}{\type_1} \\
}
{\teval{\classtable}{\Gamma}{\emethcall{\val_r}{m}{\val}}{\type_2}}[\tappBase]
$$

Thus, by inverting \tappBase, there must exist $\type_{arg}$, $\type_1$, $\type_2$ such that:

\begin{itemize}
\item[(8)] \teval{\classtable}{\Gamma}{\val}{\type_{arg}} 
\item[(9)] \cteval{\classtable}{A.m}{\mthtype{\type_1}{\type_2}} 
\item[(10)] \subtype{\type_{arg}}{\type_1} 
\end{itemize}

Now, let $\Delta = [\var \mapsto \type_1, \eself \mapsto A_{rec}]$.
Also, let $\typestack' = \stacktop{\tstackelement{\Gamma}{\type_2}{\type_C}}{\typestack}$.
By (1b), (1c), (6), (8), and (9), we know there exists $\type'$ such that \teval{\classtable}{[\eself \mapsto A_{rec}, \var \mapsto \type_1]}{\expr}{\type'} where \subtype{\type'}{\type_2}, 
which gives us (a); that is to say, the body of the method type checks as expected.
(b) holds by construction of $\typestack'$. (c) holds 
by construction of $\Delta$.
(e) holds trivially since $\stack\neq\stack'$.

Finally, we show (d). By (4) we have \consistent{\Gamma}{\dynenv}, and as noted above we have \teval{\classtable}{\Gamma\subst{\square}{\type_2}}{\context}{\type_C}. This gives us \consistent{\tstackelement{\Gamma}{\type_2}{\type_C}}{\dstackelement{\dynenv}{\context}}. By (3) we have \subtype{\type_C}{\typestack}, and by (5) we have \consistent{\typestack}{\stack}. Putting this all together, by the definition of stack consistency, this gives us \consistent{\stacktop{\tstackelement{\Gamma_C}{\type_2}{\type_C}}{\typestack}}{\stacktop{\dstackelement{\dynenv}{\context}}{\stack}}, which is (d).

\item Case (\reRet). By assumption we have
\begin{itemize}
\item[(1)] \eeval{\dyn{\dynenv'}{\val}{\stacktop{\dstackelement{\dynenv}{\context}}{\stack}}}{\dyn{\dynenv}{\context[\val]}{\stack}}
\item[(2)] \teval{\classtable}{\Gamma}{\val}{\type}
\item[(3)] \subtype{\type}{\stacktop{\tstackelement{\Gamma_C}{\type_C}{\type_C'}}{\typestack}} 
\item[(4)] \consistent{\Gamma}{\dynenv'}
\item[(5)] \consistent{\stacktop{\tstackelement{\Gamma_C}{\type_C}{\type_C'}}{\typestack}}{\stacktop{\dstackelement{\dynenv}{\context}}{\stack}}
\item[(6)] \valid(\classtable)
\end{itemize}
Let $\Delta = \Gamma_C$
and $\typestack' = \typestack$. 
By (5) we have \teval{\classtable}{\Gamma_C\subst{\square}{\type_C}}{\context}{\type_C'}, and by (3) we have \subtype{\type}{\type_C}. Then by these, (2), and the substitution lemma~\ref{lemma:substitution}, we have that \teval{\classtable}{\Gamma_C}{\context[\val]}{\type_C''} where \subtype{\type_C''}{\type_C'}; this also means that
\teval{\classtable}{\Delta}{\context[\val]}{\type_C''}.
%


 Let $\type' = \type_C''$ and we have (a). By (5) and the definition of stack consistency we have \subtype{\type_C'}{\typestack}, and since \subtype{\type_C''}{\type_C'}, we have (b). 
Finally, (c) and (d) hold by (5), and (e) holds trivially since $\stack\neq\stack'$. 
	
\end{itemize}
\end{proof}

Before proving progress, we introduce one assumption and one lemma:

\begin{assumption}[Library Method Termination]\label{assumption:builtin-termination}
For any $A$, $m$, $\val_1$, and $\val_2$ where $A.m\in\builtin$, \call{A.m}{\val_1}{\val_2} will terminate
and return a value. 
\end{assumption}

This assumption is necessary to prove progress since we do not have the
mechanism to refer to the step-by-step evaluation of library methods. A straightforward expansion of \corelang
would allow us to do so, but we omit such an expansion here for simplicity. 

\begin{lemma}[Proper context]\label{lemma:proper-context}
For any expression \expr, if $\expr = \context[\expr']$ for some $\context, \expr'$, then there exists a proper context $\context_P$ and proper subexpression $\expr_P$ such that $\expr=\context_P[\expr_P]$, $\expr_P\neq\val$ for any value \val,  and $\not\exists \context', \expr''$ such that $\expr_P = \context'[\expr'']$.
\end{lemma}

Note that it is still possible in Lemma~\ref{lemma:proper-context} that $C_P = C$ and $\expr_P = \expr$. The high-level idea behind the proof of Lemma~\ref{lemma:proper-context} is that we construct the proper context by recursively pushing the hole $[]$ deeper into subcontexts while possible. With these, we can proceed with defining and proving progress: 

\begin{theorem}[Progress]
\textit{If}
\begin{itemize}
\item[(1)] \teval{\classtable}{\Gamma}{\expr}{\type}
\item[(2)] \subtype{\type}{\typestack}
\item[(3)] \consistent{\Gamma}{\dynenv}
\item[(4)] \consistent{\typestack}{\stack}
\item[(5)] \valid(\classtable)
\end{itemize}
\textit{then one of the following holds}
\begin{itemize}
\item[1.] \expr is a value
\item[2.] There exists  $\dynenv'$, $\expr'$, $\stack'$ such that \eeval{\dyn{\dynenv}{\expr}{\stack}}{\dyn{\dynenv'}{\expr'}{\stack'}}
\item[3.] \eeval{\dyn{\dynenv}{\expr}{\stack}}{blame}
\end{itemize}
\end{theorem}

\begin{proof}
By induction on \expr.
\begin{itemize}
\item Case \vnil, \vtrue, \vfalse, \vinst{A}, or \type. \expr is a value.
\item Case \eself. By assumption (1) and (\tSelf), we know $\eself\in dom(\Gamma)$. By (3), this means $\eself\in dom(\dynenv)$. This means rule (ESelf) can be applied, so we can take a step.
\item Cases \var, \trec. Similar to \eself.
\item Case \eq{\expr_1}{\expr_2} We split into cases on whether or not $\expr_1$, $\expr_2$ are values, for any $\val_1, \val_2$:
\begin{itemize}
	\item $\expr_1 \equiv \val_1$ and $\expr_2 \equiv \val_2$. This case is trivial since one of (\reEEqTrue) or (\reEEqFalse) will always apply.
	\item $\expr_1\not\equiv\val_1$. Then $\expr \equiv \context[\expr_1]$ with $\context \equiv \eq{[]}{\expr_2}$. By Lemma~\ref{lemma:proper-context}, there exists a proper context $\context_P$ and proper subexpression $\expr_P$ such that $\expr \equiv \context_P[\expr_P]$. If $\expr_P\not\equiv\emethcall{\val_r}{m}{\val}$ for any $\val_r, A, m, \val$ where $\typeof(\val_r)=A_r \land A_r.m\in\userdef$, then by the inductive hypothesis 
there exists $\dynenv', \expr_P', \stack'$ such that \eeval{\dyn{\dynenv}{\expr_P}{\stack}}{\dyn{\dynenv'}{\expr_P'}{\stack'}} (or such that we step to blame, in which case \expr steps to blame and we are done). See the (\reContext) case of the preservation proof for a discussion of how to satisfy the premises of the inductive hypothesis, since the premises here a subset of those of preservation. By construction of the proper subexpression as specified in Lemma~\ref{lemma:proper-context}, $\not\exists \context', \expr''$ such that $\expr_P = \context'[\expr'']$, and $\expr_P\neq\val$ for any value \val. This satisfies all the premises of (\reContext), therefore this rule would apply to $\context_P[\expr_P]$ to take a step.
	
 Otherwise, $\expr_P\equiv\emethcall{\val_r}{m}{\val}$ with $\typeof(\val_r) = A_r \land A_r.m\in\userdef$. We know that (1) must have been derived by rule (\tEq). By inversion of this rule, this means \teval{\classtable}{\Gamma}{\emethcall{\val_r}{m}{\val}}{\type_m} for some $\type_m$. 
 By definition of \typeof, it must be that $\type_m = A_r$. 
 Because $A_r.m\in\userdef$ and by definition $\userdef\cap\builtin = \emptyset$, we therefore know \teval{\classtable}{\Gamma}{\emethcall{\val_r}{m}{\val}}{\type_m} must have been derived from (\tappBase). By inversion of this rule, we know \cteval{\classtable}{A_r.m}{\mthtype{\type_1}{\type_2}} for some $\type_1$, $\type_2$. By (5), this means \defof{A_r.m}{\var}{\expr_m} for some $\expr_m$.
 If $\val_r$ is \vnil, then we return blame.
 Thus, we have satisfied all the premises of rule (\reAppUD), therefore we can apply this rule and we are done.

	\item $\expr_1 \equiv \val_1$ and $\expr_2 \not\equiv \val_2$. Then $\expr \equiv \context[\expr_2]$ with $\context\equiv \eq{\val_1}{[]}$. By a similar argument to the previous case, either (\reContext) or (\reAppUD) must apply.
\end{itemize}
\item Case \eseq{\expr_1}{\expr_2}. We split cases on whether or not $\expr_1$ is a value: 
\begin{itemize}
\item $\expr_1 \equiv \val$. Then, the evaluation rule \reSeq on the expression \eseq{\val}{\expr_2} applies to take a step.
\item Otherwise, $\expr \equiv \context[\expr_1]$ with $\context \equiv \eseq{[]}{\expr_2}$. By a similar argument to the $\eq{\expr_1}{\expr_2}$ case, either (\reContext) or (\reAppUD) must apply.
\end{itemize}
\item Case \enew{A}. Trivial.

\item Case \eif{\expr_0}{\expr_1}{\expr_2}. We split cases on the structure of $\expr_0$. 
\begin{itemize}
\item $\expr_0 \equiv \vnil$ or $\expr_0\equiv\vfalse$. The rule \reEIfFalse applies to take a step.
\item $\expr_0 \equiv \val$ where \val is not \vnil or \vfalse. The rule \reEIfTrue applies to take a step.
\item Otherwise, $\expr \equiv \context[\expr_0]$ with $\context \equiv \eif{[]}{\expr_1}{\expr_2}$. By a similar argument to the $\eq{\expr_1}{\expr_2}$ case, either (\reContext) or (\reAppUD) must apply.
\end{itemize}

\item Case \emethcall{\expr_1}{m}{\expr_2}. We split this case on the structure of $\expr_1, \expr_2$:
\begin{itemize}
\item $\expr_1\not\equiv \val$. Then $\expr \equiv \context[\expr_1]$ with $\context \equiv \emethcall{[]}{m}{\expr_2}$. By a similar argument to the $\eq{\expr_1}{\expr_2}$ case, either (\reContext) or (\reAppUD) must apply.
\item $\expr_1\equiv \val_1$ and $\expr_2\not\equiv \val_2$. Then $\expr \equiv \context[\expr_2]$ with $\context \equiv \emethcall{\expr_1}{m}{[]}$. By a similar argument to the $\eq{\expr_1}{\expr_2}$ case, either (\reContext) or (\reAppUD) must apply.

\item $\expr_1=\val_1$ and $\expr_2 = \val_2$. 
If $\val_1$ is \vnil, we step to $blame$.
Because $\expr$ is a non-checked method call, we know (1) must have
been derived by (\tappBase). By inversion of this rule, \teval{\classtable}{\Gamma}{\val_1}{A_1} for some $A_1$,
and that $A_1.m\in\userdef$ and \cteval{\classtable}{A.m}{\mthtype{\type_{in}}{\type_{out}}}.
By (7), this means \defof{A.m}{\var}{\expr_m} for some $\expr_m$. 
This satisfies all the premises of rule (\reAppUD), therefore we can apply this rule.
\end{itemize}

\item Case \echeckmeth{\type_{res}}{\expr_1}{m}{\expr_2}. We split this case on the structure of $\expr_1, \expr_2$:
\begin{itemize}
	\item $\expr_1\not\equiv \val$. Then $\expr \equiv \context[\expr_1]$ with $\context \equiv \emethcall{[]}{m}{\expr_2}$. By a similar argument to the $\eq{\expr_1}{\expr_2}$ case, either (\reContext) or (\reAppUD) must apply.
	\item $\expr_1\equiv \val_1$ and $\expr_2\not\equiv \val_2$. Then $\expr \equiv \context[\expr_2]$ with $\context \equiv \emethcall{\expr_1}{m}{[]}$. By a similar argument to the $\eq{\expr_1}{\expr_2}$ case, either (\reContext) or (\reAppUD)
	
	\item $\expr_1 = \val_1$ and $\expr_2 = \val_2$. 
	If $\val_1$ is \vnil, we return blame. 
	Because \expr is a checked method call, we know (1) must have been
	derived by (\tappComp). By inversion of this rule, we know \teval{\classtable}{\Gamma}{\val_1}{A_1} for some $A_1$, 
	where $A_1.m\in\builtin$.
	Let $\val_{res}$ = \call{A.m}{\val_1}{\val_2}; by Assumption~\ref{assumption:builtin-termination} we know that this call will terminate and return a value. Then, if $\typeof(\val_{res})$ is not a
	subtype of $\type_{res}$, we will return blame. Otherwise, we will have satisfied all the preconditions of (\reAppBI),
	therefore we can apply this rule and take a step.
	
\end{itemize}

\end{itemize}
\end{proof}

We now introduce our theorem of soundness of the type checking judgment. 

\begin{theorem}[Soundness of Type Checking]
\label{theorem:tc-sound}
\textit{If \valid(\classtable) and \tcheck{\classtable}{\emptyset}{\expr}{\expr'}{\type_C} and \teval{\classtable}{\emptyset}{\expr'}{\type}, then either $\expr'$ reduces to a value, $\expr'$ reduces to blame, or $\expr '$ does not terminate.}
\end{theorem}

\begin{proof}
Let $\Gamma = \emptyset$, $\dynenv = \emptyset$, $\stack = \stacktop{\dstackelement{\emptyset}{\square}}{\cdot}$, and $\typestack = \stacktop{\tstackelement{\Gamma}{\type}{\type}}{\cdot}$. By construction we have \subtype{\type}{\typestack}, \consistent{\Gamma}{\dynenv}, and \consistent{\typestack}{\stack}. Thus, we satisfy the preconditions of progress and preservation, and soundness holds by standard argument.
\end{proof}

With this soundness theorem, it is straightforward to extend soundness to the check insertion rules.
We make use of a lemma that states that the type assigned by the check insertion rules will be equivalent
to the type assigned by the type checking rules.

\begin{lemma}
\label{lemma:check-tc-eq}
	\tcheck{\classtable}{\Gamma}{\expr}{\expr'}{\type_C} and \teval{\classtable}{\Gamma}{\expr'}{\type} if and only if
	\tcheck{\classtable}{\Gamma}{\expr}{\expr'}{\type}.
\end{lemma}

\begin{proof}
Straightforward by induction on check insertion rules \tcheck{\classtable}{\Gamma}{\expr}{\expr'}{\type}.	
\end{proof}

\begin{theorem}[Soundness of Check Insertion]
If \valid(\classtable) and \tcheck{\classtable}{\emptyset}{\expr}{\expr'}{\type} then either $\expr'$ reduces to a value, $\expr'$ reduces to blame, or $\expr '$ does not terminate.
\end{theorem}

\begin{proof}
By Theorem~\ref{theorem:tc-sound} and Lemma~\ref{lemma:check-tc-eq}.	
\end{proof}

\begin{figure*}
	\judgementHead{Program type checking rules}{\progtyp{\classtable}{P}}
	$$
	\inference{\cteval{\classtable}{A.m}{\mthtype{\type_1}{\type_2}} &&
		\teval{\classtable}{[\eself\mapsto A, \var \mapsto \type_1]}{\expr}{\type_2'} \\
		\subtype{\type_2'}{\type_2} 
	}
	{\progtyp{\classtable}{\pdef{A.m}{\var}{\mthtype{\type_1}{\type_2}}{\expr}}}[\tPDef]
	\quad
	\inference{\cteval{\classtable}{A.m}{\methtype}}
	{\progtyp{\classtable}{\bidef{A.m}{\var}{\methtype}}}[\tPBuiltIn]
	$$
	\trspace
	$$
	\inference{\progtyp{\classtable}{P_1} && \progtyp{\classtable}{P_2}}
	{\progtyp{\classtable}{P_1, P_2}}[\tPSeq]
	$$
	
	\caption{Program type checking rules.}
	\label{fig:progtyp}
	
\end{figure*}

\subsection{Program Type Checking and Class Table Construction}
\label{subsec:prog-type-check}

The rules for type checking a program are given in Figure~\ref{fig:progtyp}. They
rely on using a class table. We omit a formal definition of the class table construction here as it is straightforward.
Informally: to construct a class table, traverse a program, adding type annotations from definitions and declarations
as you go. We can then check that \progtyp{\classtable}{P} to ensure that the program $P$ type checks under the
constructed class table. If \progtyp{\classtable}{P}, and the appropriate subtyping relations hold among methods
of subclasses, we can conclude that $valid(\classtable)$ according to definition~\ref{def:class-table-validity}.

\end{document}